\documentclass[smallextended]{svjour3hacked} %%% ARXIV VERSION

\smartqed  % flush right qed marks, e.g. at end of proof
\usepackage{graphicx}

\usepackage[usenames,dvipsnames]{xcolor}

\usepackage{amssymb}
\usepackage{amsmath}
\usepackage{mathtools}
\usepackage{tikz}

\usepackage{microtype}
\usepackage{hyperref}

\usepackage{booktabs}

\usepackage{caption}
\usepackage{subcaption}

\renewcommand{\P}{\ensuremath{\mtext{P}}}
\newcommand{\NP}{\ensuremath{\mtext{NP}}}
\newcommand{\PSPACE}{\ensuremath{\mtext{PSPACE}}}
\newcommand{\para}[1]{\ensuremath{\mtext{para-#1}}}

\newcommand{\Sat}{\ensuremath{\mtext{\sc Sat}}}

\newcommand{\QSat}[1]{\ensuremath{\mtext{\sc QSat}_{#1}}}

\newcommand{\DeltaP}[1]{\ensuremath{\Delta^{\mtext{p}}_{#1}}}

\newcommand{\SigmaP}[1]{\ensuremath{\Sigma^{\mtext{p}}_{#1}}}

\newcommand{\AAA}{\mathcal{A}}

\newcommand{\EEE}{\mathcal{E}}
\newcommand{\LLL}{\mathcal{L}}
\newcommand{\MMM}{\mathcal{M}}
\newcommand{\NNN}{\mathcal{N}}
\newcommand{\PPP}{\mathcal{P}}

\newcommand{\SB}{\{\,}%
\newcommand{\SM}{\;{:}\;}%
\newcommand{\SE}{\,\}}%
\newcommand{\SBs}{\{}%
\newcommand{\SEs}{\}}%

\newcommand{\mtext}[1]{\text{\normalfont #1}}
\newcommand{\ccfont}[1]{\textnormal{\textsf{#1}}}

\newcommand{\citeauthor}[1]{\textcolor{red}{[!!]}}
\newcommand{\citeyear}[1]{\cite{#1}}

\newenvironment{myquote}{\list{}{\leftmargin=\parindent\rightmargin=0in\topsep=3pt}\item[]}{\endlist}

\journalname{Journal of Logic, Language and Information}
\begin{document}

\title{On the Computational Complexity of Model Checking for Dynamic Epistemic Logic with S5 Models}

\titlerunning{On the Complexity of Model Checking for DEL with S5 Models}        % if too long for running head

\author{Ronald de Haan
\and
Iris van de Pol\thanks{Iris van de Pol was supported by Gravitation Grant 024.001.006 of the Language in Interaction Consortium from the Dutch Research Council (NWO).}}


\institute{Ronald de Haan \at
              Institute for Logic, Language and Computation \\
              University of Amsterdam \\
              Amsterdam, the Netherlands \\
              \email{me@ronalddehaan.eu}
           \and
           Iris van de Pol \at
              Institute for Logic, Language and Computation \\
              University of Amsterdam \\
              Amsterdam, the Netherlands \\
              \email{i.p.a.vandepol@uva.nl}
}

%\date{Received: date / Accepted: date} %%% JOURNAL VERSION
\date{} %%% ARXIV VERSION

\maketitle

\begin{abstract}
Dynamic epistemic logic (DEL) is a logical framework for
representing and reasoning about
knowledge change for multiple agents.
An important computational task in this framework
is the model checking problem, which has been shown
to be \PSPACE{}-hard even for S5 models
and two agents{---%
in the presence of other features, such as multi-pointed models}.
We answer open questions in the literature about the complexity
of this problem in more restricted settings.
We provide a detailed complexity analysis
of the model checking problem for DEL,
where we consider various combinations of restrictions,
such as the number of agents,
whether the models are single-pointed or multi-pointed,
and whether postconditions are allowed in the updates.
In particular, we show that the problem is already \PSPACE{}-hard
in (1)~the case of one agent, multi-pointed S5
models, and no postconditions, and (2)~the case of two agents,
only single-pointed S5 models, and no postconditions.
In addition, we study the setting where only
semi-private announcements are allowed as updates. 
We show that for this case the problem
is already \PSPACE{}-hard when restricted to two agents
and three propositional variables.
{
The results that we obtain in this paper help outline the exact boundaries
of the restricted settings for which the model checking problem for
DEL is computationally tractable.
%They also help indicate which algorithmic techniques could work efficiently
%in practice for solving the model checking problem
%for the various combinations of features of the DEL model.
}
\end{abstract}

%%%
%%% INTRODUCTION
%%%
\section{Introduction}

Dynamic epistemic logic (or DEL, for short) is a logical framework for
representing and reasoning about
knowledge (and belief) change for multiple agents.
This framework has applications in philosophy, cognitive science,
computer science and artificial intelligence
(see, e.g.,~\cite{VanBenthem06,BolanderAndersen11,%
VanDitmarschVanderHoekKooi08,FaginHalpernMosesVardi95,%
MeyerVanderHoek95,Verbrugge09}).
For instance, reasoning about information and knowledge change
is an important topic for multi-agent and distributed systems
\cite{VanderHoekWooldridge08}.

DEL is a very general and expressive framework,
but many settings where the framework is used allow strong
restrictions. For instance, in the context of reasoning about
knowledge, the semantic models
for the logic are often restricted to models
that contain only equivalence relations
(also called S5 models).

For many of the applications of DEL, computational and algorithmic
aspects of the framework are highly relevant.
It is important to study the complexity of computational
problems associated with the logic to determine to what
extent it can be used in practical settings,
and what algorithmic approaches are best suited to solve
these problems.
One important computational task is the problem of model checking,
where the question is to decide whether a formula
is true in a model.

%%VanEijckSchwarzentruber14
%There has been relatively little research on the computational
%complexity of the model checking problem for DEL.
%For a restricted fragment of DEL,
%known as public announcement logic
%\cite{BaltagMossSolecki98,BaltagMossSolecki99,Plaza89},
%the model checking problem is polynomial-time solvable
%\cite{VanBenthemVanEijckKooi06,KooiVanBenthem04}.
%Recently, Aucher and Schwarzentruber showed that the problem of
%DEL model checking, in its general form,
%is \PSPACE{}-hard \cite{AucherSchwarzentruber13}. %
%In fact, they showed \PSPACE{}-completeness for the fragment
%of dynamic epistemic logic without postconditions.
%However, their hardness proof does not apply for the restriction
%to S5 models or for the restriction to single-pointed update models.
%They left it as an open question what the complexity of DEL model
%checking is with these restrictions.
%This open question was answered recently with a \PSPACE{}-hardness
%proof for the restricted case where all models are S5 models
%and single-pointed \cite{VandePol15,VandePolVanRooijSzymanik18}.
%%
%It remained open whether this \PSPACE{}-hardness result extends
%to more restrictive settings (e.g., only two agents and no
%multi-pointed models).

The complexity of the model checking problem for DEL
has been a topic of investigation in the literature.
For a restricted fragment of DEL,
known as public announcement logic
\cite{BaltagMossSolecki98,BaltagMossSolecki99,Plaza89},
the model checking problem is polynomial-time solvable
\cite{VanBenthemVanEijckKooi06,KooiVanBenthem04}.
The problem of DEL model checking, in its general form,
has been shown to be \PSPACE{}-complete
\cite{AucherSchwarzentruber13,VanEijckSchwarzentruber14},
even in the case of two agents and S5 models.
However, these hardness proofs crucially depend on the use of multi-pointed
models, and therefore do not apply for the case where
the problem is restricted to single-pointed S5 models.
This open question was answered with a \PSPACE{}-hardness
proof for the restricted case where all models are 
single-pointed S5 models, 
%S5 models
%and single-pointed --
but where the number of agents is unbounded \cite{VandePol15,VandePolVanRooijSzymanik18}.
It remained open whether these \PSPACE{}-hardness results extend
to more restrictive settings
%(e.g., only two agents, S5 models and only
%single-pointed update models).
%(e.g., combining all of these restrictions).
(e.g., only two agents and single-pointed S5 models).

{
In this paper, we investigate to what extent these \PSPACE{}-hardness
results hold for more demanding combinations of restrictions.
In other words, we study the exact boundaries between
(A)~the combinations of restrictions that lead to the model checking
problem being polynomial-time solvable,
and (B)~the combinations of restrictions for which the model checking
problem is computationally intractable.
Various examples of restrictions have been found that fit in either~(A)
or~(B), but no structural investigation has been done on the exact
boundaries between these two areas.
Investigating these exact boundaries is useful and relevant, for example,
for the development of (implemented) algorithms for the model checking
problem for DEL---we discuss this relevance in more detail
in Section~\ref{sec:discussion}.
}

\paragraph{Other related work}
Various topics related to DEL
model checking have been studied in the literature.
%\red{Besides the complexity of the model checking problem for DEL, 
%other related problems and algorithms have been studied in the literature.}
For (several restricted variants of)
a knowledge update framework based on epistemic logic,
the computational complexity of the model checking problem
has been investigated \cite{BaralZhang05}.
Other related work includes implementations of algorithms
for DEL model checking
\cite{VanBenthemVanEijckGattingerSu18,VanEijck07}.
Additionally, research has been done on the complexity of the
satisfiability problem for (fragments of) DEL
\cite{AucherSchwarzentruber13,Lutz2006}.

\paragraph{{Results and contributions}}
In this paper
we provide a detailed computational complexity analysis
of the model checking problem for DEL, restricted to S5 models.
We consider various different restricted settings of this problem.

For the case of arbitrary event models, we have the following
results.
\begin{itemize}
  \item {We make the following folklore result explicit:}
    that the problem is polynomial-time solvable
    in the case of a single agent and single-pointed S5 models
    without postconditions
    (Proposition~\ref{prop:ptime}).
  \item We show that a similar restriction (single agent and
    single-pointed S5 models) where postconditions are allowed already
    leads to \DeltaP{2}-hardness (Theorem~\ref{thm:delta2}).
  \item When multi-pointed event models are allowed, we show that the
    problem is \PSPACE{}-hard even for the case of
    a single agent and S5 models without postconditions
    (Theorem~\ref{thm:pspace-hardness2}).
  \item For the case where there are two agents, we show that the problem
    is already \PSPACE{}-hard when restricted to single-pointed
    S5 models without postconditions
    and with only three propositional variables
    (Theorem~\ref{thm:pspace-hardness}).
\end{itemize}
An overview of the complexity results for arbitrary event
models can be found in Table~\ref{table:results1}.
{
These results outline the boundaries of the tractable setting
of the folklore results pinpointed in Proposition~\ref{prop:ptime}---%
they indicate that relaxing any of the three elements
of the condition (i.e., a single agent, single-pointed models,
and no postconditions) results in computational hardness.
}

Additionally, we consider the setting where instead of arbitrary
event models, only semi-private announcements can be used---%
this is a restricted class of event models.
In this setting, the problem is known to be \PSPACE{}-hard,
when an arbitrary number of agents is allowed
(i.e., when the number of agents is part of the problem input)
\cite[Theorem~4]{VandePol15}.
\begin{itemize}
  \item We show that the problem is already \PSPACE{}-hard
    in the case where there are only two agents
    and only three propositional variables
    (Theorem~\ref{thm:semi-private-pspace-hardness}).%
\footnote{We would like to point out that
Theorem~\ref{thm:semi-private-pspace-hardness} is a stronger
result than Theorem~\ref{thm:pspace-hardness}---%
Theorem~\ref{thm:semi-private-pspace-hardness}
implies the result of Theorem~\ref{thm:pspace-hardness}.
%{We present it as two separate results for two reasons:
%(i)~the result of Theorem~\ref{thm:pspace-hardness} by itself already
%provides new and interesting information \textcolor{red}{[todo]} that is relevant for the
%algorithmic study of the model checking problem for DEL---
%see~Section~\ref{sec:discussion}---%
%and (ii)~the proof of Theorem~\ref{thm:pspace-hardness} acts as a
%stepping stone for proving Theorem~\ref{thm:semi-private-pspace-hardness}.}}
{We present it as two separate results because
the proof of Theorem~\ref{thm:pspace-hardness} acts as a
stepping stone for proving Theorem~\ref{thm:semi-private-pspace-hardness}---%
the proof of Theorem~\ref{thm:pspace-hardness} is useful for
understanding the elaborate proof of Theorem~\ref{thm:semi-private-pspace-hardness}.}}
\end{itemize}

%%%
\begin{table}[h!]
  \centering

  \begin{tabular}{@{}l@{\quad}l@{\quad}l@{\quad}r@{\quad} r@{}}
    \toprule
    \# agents & single- or multi-pointed & postconditions & complexity & \\
    \midrule
    $1$ & single & no & in \P{} & (Proposition~\ref{prop:ptime}) \\
    $1$ & single & yes & \DeltaP{2}-hard & (Theorem~\ref{thm:delta2}) \\
    $1$ & multi & no / yes & \PSPACE{}-complete & (Theorem~\ref{thm:pspace-hardness2}) \\
    $\geq 2$ & single / multi & no / yes & \PSPACE{}-complete & (Theorem~\ref{thm:pspace-hardness}) \\
    \bottomrule
  \end{tabular}
  
  \caption{Complexity results for the model checking problem
    for DEL with S5 models and S5 event models.}
  \label{table:results1}
\end{table}

{
\paragraph{Interpretation of the results}
The results that we obtain in the paper contribute to our
understanding of the computational complexity of the model checking problem
for DEL.
In particular, our results form a useful step towards a better comprehension
of how the various elements of the framework of DEL contribute to the
computational costs of the model checking problem.
{
For example, the hardness result of Theorem~\ref{thm:pspace-hardness}
indicates that introducing a second agent---even when severely restricting
several other aspects of the problem---already leads to a problem
that in the worst case is as hard as the general, unrestricted problem.
}
This improved insight can---in future work---be used to develop
(implemented) algorithms for DEL model checking that work more efficiently
in different settings and for different applications.
We discuss the relevance and significance of our results in more
detail in Section~\ref{sec:discussion}.
}

\paragraph{Roadmap}
We begin in Section~\ref{sec:preliminaries}
with reviewing basic notions and notation from
dynamic epistemic logic and complexity theory.
Then, in Section~\ref{sec:updates-with-event-models},
we present the
complexity results for the various settings that involve
updates with (arbitrary) event models.
In Section~\ref{sec:semi-private-announcements}, we present
our \PSPACE{}-hardness proof
for the setting of semi-private announcements.
{We discuss the relevance and significance of our
results for the computational and algorithmic study of the model
checking problem for DEL in Section~\ref{sec:discussion}.}
Finally, we conclude and suggest directions for future
research in Section~\ref{sec:conclusion}.

%%%
%%% PRELIMINARIES
%%%
\section{Preliminaries}
\label{sec:preliminaries}

We briefly review some basic notions from dynamic epistemic logic
and complexity theory that are required for the complexity results
that we present in this paper.

\subsection{Dynamic Epistemic Logic}

We begin by reviewing the syntax and semantics of dynamic epistemic logic.
We consider a version of this logic that is often considered in the
literature (e.g., by Van Ditmarsch, Van der Hoek and Kooi
\cite{VanDitmarschVanderHoekKooi08}).
After describing the logic that we consider in this paper, we briefly
relate it to other variants of dynamic epistemic logic that have been
considered in the literature.

We fix a countable set~$\PPP$ of propositional variables,
and a finite set~$\AAA$ of agents.
We begin with introducing the basic language of epistemic logic,
and its semantics.
The semantics of epistemic logic is based on a type of (Kripke)
structures called \emph{epistemic models}.
Epistemic models are structures that are used to represent
the agents' knowledge about the world and about the
other agents' knowledge.

\begin{definition}[Epistemic models]
An \emph{epistemic mo\-del} is a tuple~$\MMM = (W,R,V)$,
where~$W$ is a non-empty set of worlds,~$R$ maps each
agent~$a \in \AAA$ to a relation~$R_a \subseteq W \times W$,
and~$V : \PPP \rightarrow 2^{W}$ is a function called a
valuation.
By a slight abuse of notation, we write~$w \in \MMM$
for~$w \in W$.
We also write~$v \in R_a(w)$ for~$vR_aw$.
A \emph{single-pointed model} is a pair~$(\MMM,w)$ consisting
of an epistemic model~$\MMM$ and a \emph{designated}
(or \emph{pointed}) \emph{world~$w \in \MMM$}.
A \emph{multi-pointed model} is a pair~$(\MMM,W_d)$ consisting
of an epistemic model~$\MMM$ and a subset~$W_d$
of designated worlds.
\end{definition}

\begin{definition}[Basic epistemic language]
The language~$\LLL_{\mtext{EL}}$ of epistemic logic is defined
as the set of formulas~$\varphi$ defined inductively
as follows, where~$p$ ranges over~$\PPP$
and~$a$ ranges over~$\AAA$:
\[ \varphi \Coloneqq p\ |\ \neg\varphi\ |\ (\varphi \wedge \varphi)\ |\ 
  K_{a} \varphi. \]
The formula~$\bot$ is an abbreviation for~$p \wedge \neg p$,
and the formula~$\top$ is an abbreviation for~$\neg\bot$.
A formula of the form~$(\varphi_1 \vee \varphi_2)$
abbreviates~$\neg(\neg \varphi_1 \wedge \neg \varphi_2)$,
and a formula of the form~$(\varphi_1 \rightarrow \varphi_2)$
abbreviates~$(\neg \varphi_1 \vee \varphi_2)$.
Moreover, a formula of the form~$\hat{K}_a \varphi$ is an abbreviation
for~$\neg K_a \neg \varphi$.
We call formulas of the form~$p$ or~$\neg p$ literals.
We denote the set of all literals by~$\mtext{Lit}$.
\end{definition}

Intuitively, the formula~$K_a \varphi$ expresses
that `agent~$a$ knows
that~$\varphi$ holds in the current situation.'
Next, we define when a formula in the basic epistemic
language is true in a world of an epistemic model.

\begin{definition}[Truth conditions for~$\LLL_{\mtext{EL}}$]
\label{def:truthconditions1}
Given an epistemic model~$\MMM = (W,R,V)$,
we inductively define
the relation~$\models\ \subseteq W \times \LLL_{\mtext{EL}}$
as follows.
For all~$w \in W$:\\[5pt]
\begin{tabular}{l l l}
$\MMM,w \models p$ & iff & $w \in V(p)$ \\
$\MMM,w \models \neg \varphi$ & iff & not $\MMM,w \models \varphi$ \\
$\MMM,w \models \varphi_1 \wedge \varphi_2$ &
  iff & both~$\MMM,w \models \varphi_1$
  and~$\MMM,w \models \varphi_2$ \\
$\MMM,w \models K_a \varphi$ & iff & for all~$v \in R_a(w)$,
  it holds that~$\MMM,v \models \varphi$ \\[5pt]
\end{tabular}
\smallskip

\noindent The statement~$\MMM,w \models \varphi$ expresses that the
formula~$\varphi$ is true in world~$w$ in the model~$\MMM$.
\end{definition}

The framework of dynamic epistemic logic extends the basic epistemic
logic with a notion of updates, that are based on another type of
structures: \emph{event models}.
These are used to represent the effects of an event on the
world and the knowledge of the agents.
{
The notion of event models that we use in this paper involves
postconditions---to bring about changes in the factual state of the
world. Event models with postconditions have been studied and used
in the literature on dynamic epistemic logic
and epistemic planning
(see, e.g.,~\cite{BolanderAndersen11,VanDitmarschKooi06}).
}

\begin{definition}[Event models]
An \emph{event model} is a tuple~$\EEE = (E,S,\ccfont{pre},\ccfont{post})$,
where~$E$ is a non-empty and finite set of possible events,~$S$
maps each agent~$a \in \AAA$ to a
relation~$S_a \subseteq E \times E$,~$\ccfont{pre} : E \rightarrow
\LLL_{\mtext{EL}}$
is a function that maps each event to a precondition
expressed in the epistemic language,
and~$\ccfont{post} : E \rightarrow 2^{\mtext{Lit}}$
is a function that maps
each event to a set of literals (not containing complementary
literals)%
\footnote{Alternatively, one can define postconditions using a
function~$\ccfont{post}: E \times \PPP \rightarrow \LLL_{\mtext{EL}}$,
(see, e.g.,~\cite{VanDitmarschKooi06}).
The complexity results in this paper also hold
when this alternative definition is used.}.
For convenience, we write~$\top$ to denote an
empty postcondition.
By a slight abuse of notation, we write~$e \in \EEE$ for~$e \in E$.
A \emph{single-pointed event model} is a pair~$(\EEE,e)$ consisting
of an event model~$\EEE$ and a \emph{designated} (or \emph{pointed})
\emph{event~$e \in \EEE$}.
A \emph{multi-pointed event model} is a pair~$(\EEE,E_d)$ consisting
of an event model~$\EEE$ and a subset~$E_d \subseteq E$ of
designated events.
\end{definition}

The language of dynamic epistemic logic extends the basic epistemic
language with update modalities.

\begin{definition}[Dynamic epistemic language]
The language~$\LLL_{\mtext{DEL}}$ of dynamic epistemic logic is defined
as the set of formulas~$\varphi$ defined inductively as follows:
\[ \varphi \Coloneqq p\ |\ \neg\varphi\ |\ \varphi \wedge \varphi\ |\ 
  K_{a} \varphi\ |\ [\EEE,e]\varphi\ |\ [\EEE,E_d]\varphi, \]
where~$p$ ranges over~$\PPP$ and~$a$ ranges over~$\AAA$,
and where~$(\EEE,e)$ and~$(\EEE,E_d)$ are single- and multi-pointed
event models, respectively.
A formula of the form~$\langle \EEE,e \rangle \varphi$ is
an abbreviation for~$\neg [\EEE,e] \neg \varphi$;
we use a similar abbreviation~$\langle \EEE,E_d \rangle \varphi$
for updates with multi-pointed event models.
\end{definition}

The effect of these event models is defined using the following notion
of product update.

\begin{definition}[Product update]
Let~$\MMM = (W,R,V)$ be an epistemic model
and let~$\EEE = (E,S,\ccfont{pre},\ccfont{post})$ be an event model.
The \emph{product update of~$\MMM$ by~$\EEE$} is the epistemic
model~$\MMM \otimes \EEE = (W',R',V')$ defined as follows,
where~$p$ ranges over~$\PPP$ and~$a$ ranges over~$\AAA$:
\[ \begin{array}{r l}
  W' =&
    \SB (w,e) \in W \times E \SM \MMM,w \models \ccfont{pre}(e) \SE \\[3pt]
  R'_a =&
    \SB ((w,e),(w',e')) \in W' \times W' \SM
    wR_aw' \mtext{ and } eS_ae' \SE \\[3pt]
  V'(p) =&
    \SB (w,e) \in W' \SM w \in V(p) \mtext{ and }
      \neg p \not\in \ccfont{post}(e) \SE\ \cup \\[2pt]
   & \SB (w,e) \in W' \SM p \in \ccfont{post}(e) \SE \\
\end{array} \]
\end{definition}

%Given a single-pointed epistemic model~$(\MMM,w)$ and a single-pointed
%event model~$(\EEE,e)$, we say that~$(\EEE,e)$ is \emph{applicable}
%in~$(\MMM,w)$ if~$(w,e) \in \MMM \otimes \EEE$, i.e.,
%if~$\MMM,w \models \ccfont{pre}(e)$.
%Similarly, we say that a multi-pointed event model~$(\EEE,E_d)$ is
%applicable in~$(\MMM,w)$ if there exists some~$e \in E_d$
%such that~$(\EEE,e)$ is applicable in~$(\MMM,w)$.
%%
%Intuitively, the formula~$[\EEE,e]\varphi$ reads as `$\varphi$ will hold after
%the occurrence of the event represented by~$(\EEE,e)$', and
%the formula~$\langle \EEE,e \rangle \varphi$ reads as `the event
%represented by~$(\EEE,e)$ is applicable in the current situation, and
%after the occurrence of this event~$\varphi$ will hold.'
%
%The semantics of DEL-formulas is defined as follows.

Next, we define when a formula in the dynamic epistemic
language is true in a world of an epistemic model.

\begin{definition}[Truth conditions for~$\LLL_{\mtext{DEL}}$]
\label{def:truthconditions2}
Given an epistemic model~$\MMM = (W,R,V)$ and a
formula~$\varphi \in \LLL_{\mtext{DEL}}$, we inductively define
the relation~$\models\ \subseteq W \times \LLL_{\mtext{DEL}}$
as follows. For all~$w \in W$:\\[5pt]
\begin{tabular}{l l l}
$\MMM,w \models [\EEE,e] \varphi$ & iff &
  $\MMM,w \models \ccfont{pre}(e)$ implies~$\MMM \otimes \EEE, (w,e) \models \varphi$ \\[2pt]
$\MMM,w \models [\EEE,E_d] \varphi$ & iff &
  $\MMM,w \models [\EEE,e] \varphi$ for all~$e \in E_d$ \\
\end{tabular}
\medskip

\noindent The other cases are identical to Definition~\ref{def:truthconditions1}.
Again, the statement~$\MMM,w \models \varphi$ expresses that the
formula~$\varphi$ is true in state~$w$ in the model~$\MMM$.
\end{definition}

(Having defined the language~$\LLL_{\mtext{DEL}}$,
we could now also change the definition of preconditions
in event models to be functions~$\ccfont{pre} : E \rightarrow
\LLL_{\mtext{DEL}}$ mapping events to formulas in the
dynamic epistemic language~$\LLL_{\mtext{DEL}}$.
The definition of product update would work in an entirely
similar way.
All results in this paper work for either definition 
of preconditions~$\ccfont{pre}$.)

We can then define truth of a formula~$\varphi \in \LLL_{\mtext{DEL}}$
in epistemic models as follows.
A formula~$\varphi$ is true in a single-pointed epistemic model~$(\MMM,w)$
if~$\MMM,w \models \varphi$, and a formula~$\varphi$ is true in a multi-pointed
epistemic model~$(\MMM,W_d)$ if~$\MMM,w \models \varphi$ for all~$w \in W_d$.

For the purposes of representing knowledge, the relations
in epistemic models and event models are often restricted to be
equivalence relations, that is, reflexive, transitive and symmetric
(see, e.g.,~\cite{VanDitmarschVanderHoekKooi08}).
Models that satisfy these requirements are also called
\emph{S5 models}, after the axiomatic system that characterizes
this type of relations.
In the remainder of this paper, we consider only epistemic models
and event models that are S5 models.
All our hardness results hold for S5 models,
as well as for arbitrary models.

For the sake of convenience, we will often depict epistemic models
and event models graphically. We will represent
worlds with solid dots, events with solid squares,
designated worlds and events with
a circle or square around them,
valuations, preconditions and postconditions with
labels next to the dots, and
relations with labelled lines between the dots.
Since we restrict ourselves to S5 models, and thus to equivalence relations,
all relations are symmetric and it suffices to represent relations
with undirected lines.
Moreover, the reflexive relations are not represented graphically.
For a valuation of a world~$w$, we use the literals that the valuation
makes true in world~$w$ as a label,
and for the preconditions and postconditions of an event~$e$,
we use the label~$\langle \ccfont{pre}(e), \ccfont{post}(e) \rangle$.
Moreover, since all epistemic models and event models that we
consider in this paper have reflexive relations, in order not to
clutter the graphical representation of models,
we do not explicitly depict the reflexive relations.
%Moreover, in order not to clutter the graphical representation of models,
%we will often only depict a subset of the relations, whose transitive,
%reflexive closure is the intended set of relations.
For an example of an epistemic model with its
graphical representation, see Figure~\ref{fig:model0},
and for an example of an event model with its
graphical representation, see Figure~\ref{fig:model0b}.

\begin{figure}[htp!]
\begin{center}
\begin{tikzpicture}
  \tikzstyle{dnode}=[inner sep=1pt,outer sep=1pt,draw,circle,minimum width=9pt]
  \tikzstyle{nnode}=[inner sep=1pt,outer sep=1pt,circle,minimum width=9pt]
  \tikzstyle{label-edge}=[midway,fill=white, inner sep=1pt]
  % NODES
  \node[dnode, label=above:{\scriptsize $z$}, label=below:{\scriptsize $w_1$}] (w0) at (0,0) {};
  \fill (w0) circle [radius=2pt];
  \node[nnode, label=above:{\scriptsize $\neg z$}, label=below:{\scriptsize $w_2$}] (w1) at (3,0) {};
  \fill (w1) circle [radius=2pt];
  % EDGES
  \draw[-] (w0) -- (w1) node[label-edge] {\scriptsize $a$};
\end{tikzpicture}
\end{center}
\caption{The epistemic model~$(\MMM,w_1)$ for
the set~$\AAA = \SBs a, b \SEs$ of agents
and a single proposition~$z$,
where~$\MMM = (W,R,V)$,
$W = \SBs w_1,w_2 \SEs$,
$R_a = \SBs (w_1,w_1), (w_1,w_2), (w_2,w_1), (w_2,w_2) \SEs$,
$R_b = \SBs (w_1,w_1), (w_2,w_2) \SEs$,
and~$V(z) = \SBs w_1 \SEs$.}
\label{fig:model0}
\end{figure}
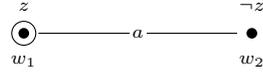

\begin{figure}[htp!]
\begin{center}
\begin{tikzpicture}
  \tikzstyle{dnode}=[inner sep=1pt,outer sep=1pt,draw,circle,minimum width=9pt]
  \tikzstyle{nnode}=[inner sep=1pt,outer sep=1pt,circle,minimum width=9pt]
  \tikzstyle{dnode'}=[inner sep=1pt,outer sep=1pt,draw,rectangle,minimum width=9pt,minimum height=9pt]
  \tikzstyle{nnode'}=[inner sep=1pt,outer sep=1pt,rectangle,minimum width=9pt,minimum height=9pt]
  \tikzstyle{label-edge}=[midway,fill=white, inner sep=1pt]
  % NODES
  \node[dnode', label=above:{\scriptsize $\langle \top, h \rangle$}, label=below:{$e_1$}] (w0) at (0,0) {};
  \fill {(w0)+(-0.07,-0.07)} rectangle (0.07,0.07);
  \node[nnode', label=above:{\scriptsize $\langle \top, \neg h \rangle$}, label=below:{$e_2$}] (w1) at (3,0) {};
  \fill {(w1)+(-0.07,-0.07)} rectangle (3.07,0.07);
  % EDGES
  \draw[-] (w0) -- (w1) node[label-edge] {\scriptsize $a$};
\end{tikzpicture}
\end{center}
\caption{The event model~$(\EEE,e_1)$ for
the set~$\AAA = \SBs a, b \SEs$ of agents
and a single proposition~$h$,
where~$\MMM = (E,S,\ccfont{pre},\ccfont{post})$,
$E = \SBs e_1,e_2 \SEs$,
$R_a = \SBs (e_1,e_1), (e_1,e_2), (e_2,e_1), (e_2,e_2) \SEs$,
$R_b = \SBs (e_1,e_1), (e_2,e_2) \SEs$,
$\ccfont{pre}(e_1) = \ccfont{pre}(e_2) = \top$,
$\ccfont{post}(e_1) = h$, and
$\ccfont{post}(e_2) = \neg h$.}
\label{fig:model0b}
\end{figure}
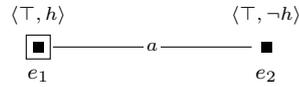

%\paragraph{{Example}}
%\todo{Come up with example; first one is simple, without postconditions,
%to model information change.}
%
%\paragraph{{Another example}}
%\todo{Come up with another example that involves ontic change (with
%postconditions); refer to previous work that uses this model.}

\paragraph{{Semi-private announcements}}
A particular type of S5 event models that has been considered
in the literature are \emph{semi-private} (or \emph{semi-public})
\emph{announcements} \cite{BaltagMoss04}.
Intuitively, a semi-private announcement publicly announces
one of two formulas~$\varphi_1,\varphi_2$ to a subset~$A$ of agents,
and to the remaining agents it publicly announces that one of the
two formulas is the case, and that the agents
in~$A$ learned which one is.
A semi-private announcement for
formulas~$\varphi_1,\varphi_2$ and a subset~$A \subseteq \AAA$ of
agents is represented by the event model in
Figure~\ref{fig:semi-private}.
%\textcolor{red}{The results that we present below
%for semi-private announcements hold even in the setting where
%the formulas~$\varphi_1$ and~$\varphi_2$ are complementary
%literals~$x$ and~$\neg x$.}
%
% one of the settings in which semi-private announcements have
% been used is that of
%In one of the settings where semi-private announcements have been
%studied, they are used as building blocks
%%\textcolor{blue}{Semi-private announcements} have also been investigated
%%in the setting of using them as building blocks
%to construct event models \cite{FrenchHalesTay14}.

To illustrate the notion of semi-private announcements, consider
the following example scenario.
There are two agents, Ayla ($a$) and Blair ($b$).
Ayla flips a coin, which lands either on heads ($h$)
or on tails ($\neg h$), and hides the result of the coin flip
from Blair.
Blair sees that the coin is flipped and that Ayla knows the result
of the coin flip, but Blair herself does not see the result of the
coin flip.
This semi-private announcement is represented by the
event model~$\EEE$ that is depicted in Figure~\ref{fig:model0b}
(in the event model depicted in Figure~\ref{fig:model0b},
the coin lands on heads).

\begin{figure}[htp!]
\begin{center}
\begin{tikzpicture}
  \tikzstyle{dnode}=[inner sep=1pt,outer sep=1pt,draw,circle,minimum width=9pt]
  \tikzstyle{nnode}=[inner sep=1pt,outer sep=1pt,circle,minimum width=9pt]
  \tikzstyle{dnode'}=[inner sep=1pt,outer sep=1pt,draw,rectangle,minimum width=9pt,minimum height=9pt]
  \tikzstyle{nnode'}=[inner sep=1pt,outer sep=1pt,rectangle,minimum width=9pt,minimum height=9pt]
  \tikzstyle{label-edge}=[midway,fill=white, inner sep=1pt]
  % NODES
  \node[nnode', label=above:{\scriptsize $\langle \varphi_1, \top \rangle$}] (w0) at (0,0) {};
  \fill {(w0)+(-0.07,-0.07)} rectangle (0.07,0.07);
  \node[nnode', label=above:{\scriptsize $\langle \varphi_2, \top \rangle$}] (w1) at (3,0) {};
  \fill {(w1)+(-0.07,-0.07)} rectangle (3.07,0.07);
  % EDGES
  \draw[-] (w0) -- (w1) node[label-edge] {\scriptsize $\AAA \backslash A$};
\end{tikzpicture}
\end{center}
\caption{A semi-private announcement for formulas~$\varphi_1,\varphi_2$
and the subset~$A \subseteq \AAA$ of agents.}
\label{fig:semi-private}
\end{figure}
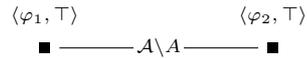

%\todo{Work out example of Ayla and Blair in more detail,
%using figures.}

\paragraph{Relations to other variants of Dynamic Epistemic Logic}
The formalism of dynamic epistemic logic that we consider is based
on the one originally introduced by
Baltag, Moss, and Solecki~%
\cite{BaltagMossSolecki98,BaltagMossSolecki99}.
Their language only considers single-pointed event models.
A few years later, Baltag and Moss~\cite{BaltagMoss04} extended
this original language to include regular operators (union,
composition and `star') for the update modalities.
The language that we consider corresponds to the variant
of their language with only the union operator.
The language presented by Van Ditmarsch et al.~%
in their textbook \cite{VanDitmarschVanderHoekKooi08}
resembles the language that we consider,
as their framework also allows the union operator
for updates, but not the composition or `star' operators.
The union operator for update modalities corresponds to allowing
multi-pointed event models.
Because it simplifies notation, we use multi-pointed models,
following the notation of other existing work \cite{BolanderAndersen11}.
%
%In addition, as preconditions of events they allow
%formulas containing update modalities as well.
%For the sake of simplicity, in our language we did not allow the
%possibility of having formulas of~$\LLL_{\mtext{DEL}}$ in the
%preconditions of event models.
%However, all the complexity results presented in this paper also work
%for this slightly extended language.
%
Additionally, the language that we consider also allows events
to have postconditions, unlike the language presented by
Van Ditmarsch et al.~%
\cite{VanDitmarschVanderHoekKooi08}.

\subsection{Computational Complexity}

Next, we review some basic notions from computational complexity
that are used in the proofs of the results that we present.
We assume the reader to be familiar with the complexity classes
\P{} and \NP{}, and with basic notions such as polynomial-time
reductions.
For more details, we refer to textbooks on computational complexity
theory (see, e.g.,~\cite{AroraBarak09}).

The class \PSPACE{} consists of all decision problems that can be
solved by an algorithm that uses a polynomial amount of space
(memory). Alternatively, one can characterize the class \PSPACE{}
as all decision problems for which there exists a polynomial-time
reduction to the problem \QSat{}, that is defined using quantified
Boolean formulas as follows.
A (fully) quantified Boolean formula (in prenex form) is a formula
of the form~$Q_1 x_1 Q_2 x_2 \dotsc Q_n x_n. \psi$,
where all~$x_i$ are propositional variables,
each~$Q_i$ is either an existential or a universal quantifier,
and~$\psi$ is a (quantifier-free) propositional formula over the
variables~$x_1,\dotsc,x_n$.
Truth for such formulas is defined in the usual way.
The problem \QSat{} consists of deciding whether a given
quantified Boolean formula is true.
Moreover, \QSat{} is \PSPACE{}-hard even when restricted
to the case where~$Q_i = \exists$ for odd~$i$ and~$Q_i = \forall$
for even~$i$.
(For the proofs of Theorems~\ref{thm:pspace-hardness2},~%
\ref{thm:pspace-hardness}
and~\ref{thm:semi-private-pspace-hardness},
we will use reductions from this restricted variant of \QSat{}.)

Additionally, one can restrict the number of quantifier alternations
occurring in quantified Boolean formulas, i.e., the number of
times where~$Q_i \neq Q_{i+1}$.
For each constant~$k \geq 1$ number of alternations,
this leads to a different complexity
class. These classes together constitute the Polynomial Hierarchy.
We consider the complexity classes~\SigmaP{k}, for each~$k \geq 1$.
The complexity class~\SigmaP{k} consists of all decision problems
for which there exists a polynomial-time reduction to the problem
\QSat{k}.
Instances of \QSat{k} are quantified Boolean formulas of the
form~$\exists x_1 \dotsc \exists x_{\ell_1} \forall x_{\ell_1+1}
\dotsc \forall x_{\ell_2} \dotsc \allowbreak Q_k x_{\ell_{k-1}+1} \dotsc Q_k x_{\ell_k}.$ $\psi$,
where~$Q_k = \exists$ if~$k$ is odd and~$Q_k = \forall$ if~$k$ is even,
where~$1 \leq \ell_{1} \leq \dotsm \leq \ell_{k}$,
and where~$\psi$ is quantifier-free.
The question is to decide whether the quantified Boolean formula is true.

The last complexity class that we consider is \DeltaP{2}.
We give a definition of this class that is based on algorithms with access
to an oracle, i.e., a black box that is able to decide certain decision
problems in a single operation.
Consider the problem \Sat{} of deciding satisfiability of a given
propositional formula.
The class \DeltaP{2} consists of all decision problems that can be solved
by a polynomial-time algorithm with access to an oracle for \Sat{}.
Alternatively, the class \DeltaP{2} consists of all decision problems for
which there exists a polynomial-time reduction to the problem
where one is given a satisfiable propositional formula~$\varphi$ over the
variables~$x_1,\dotsc,x_n$, and the question is whether the
lexicographically maximal assignment that satisfies~$\varphi$
(given the fixed ordering~$x_1 \prec \dotsm \prec x_n$)
sets variable~$x_n$ to true \cite{Krentel92}.
An assignment~$\alpha_1$ is lexicographically larger than an
assignment~$\alpha_2$ (given the ordering~$x_1 \prec \dotsm \prec x_n$)
if there exists some~$1 \leq i \leq n$ such
that~$\alpha_1(x_i) = 1$,~$\alpha_2(x_i) = 0$,
and for all~$1 \leq j \leq i$ it holds that~$\alpha_1(x_j) = \alpha_2(x_j)$.

%%%
%%% RESULTS FOR UPDATES WITH ARBITRARY S5 MODELS
%%%
\section{Results for updates with arbitrary S5 models}
\label{sec:updates-with-event-models}

In this section, we provide complexity results for the model checking problem
for DEL when arbitrary event models are allowed for the update modalities
in the formulas.
For several cases, we prove \PSPACE{}-hardness.
Since the problem was recently shown to be in \PSPACE{} for the most
general variant of dynamic epistemic logic that we consider in this
paper~\cite{AucherSchwarzentruber13,%
%VanEijckSchwarzentruber14,%
VandePol15,VandePolVanRooijSzymanik18},
these hardness results suffice to show \PSPACE{}-completeness.

%%%
%%% POLYNOMIAL-TIME SOLVABILITY
%%%
\subsection{Polynomial-time solvability}

We begin with showing polynomial-time solvability for the strongest
restriction that we consider in this paper: a single agent, single-pointed
S5 models, and no postconditions.
{
This is a result that is well-known and
can be seen as part of the folklore of the
DEL literature, but for which---to the best of our knowledge---%
no detailed proof has been published.
The high-level idea of the proof is straightforward:
even though updating the model with an event model in this setting might duplicate
a lot of words---potentially resulting in an exponential blow-up in the number
of worlds---these worlds are copies of only a small number of distinct worlds.
We can identify a representative for each of the worlds in polynomial-time,
and at each step in the recursive evaluation of the formula, we keep only these
representative worlds.
To work out an algorithm that implements this proof idea requires some
attention to algorithmic details (e.g., using the technique of dynamic programming).
{
We present the proof of Proposition~\ref{prop:ptime} in detail
to give insight into the exact algorithmic details involved in the proof---%
providing a precise recipe that can be used to implement the algorithm.
}}

\begin{proposition}
\label{prop:ptime}
The model checking problem for DEL with S5 models
is poly\-nomial-time solvable
when restricted to instances with a single agent 
and only single-pointed event models
without postconditions.
\end{proposition}
\begin{proof}
We describe a polynomial-time algorithm that solves the problem.
The main idea behind this algorithm is the following.
Even though the updates might cause an exponential blow-up in
the number of worlds in the model, in this restricted setting,
we only need to remember a small number of these worlds.

Concretely, since there is only a single agent~$a$,
and since there is only a single designated world~$w_0$,
we only need to remember the set of worlds that are connected
with an $a$-relation to the designated world~$w_0$.
Moreover, among these worlds, we can merge those with
an identical valuation.
Since the event models contain no postconditions, this
(contracted) set of worlds can only decrease after updates,
i.e., updates can only remove worlds from this set.

Formally, we can describe this argument as follows.
Let~$(\MMM,w_0)$ be a single-pointed S5 epistemic model with
one agent,
and let~$(\EEE,e_0)$ be a single-pointed S5 event model with 
one agent and no postconditions.
Then~$\MMM \otimes \EEE$ is bisimilar to a submodel~$\MMM'$
of~$\MMM$, that is, to some~$\MMM'$ that can be obtained from~$\MMM$ by
removing some worlds.
Specifically, let~$W'$ be the set of worlds in~$\MMM$ that are
$a$-accessible from~$w_0$, and let~$E'$ be the set of events
in~$\EEE$ that are $a$-accessible from~$e_0$.
Then, let~$W'' \subseteq W'$ be the subset of worlds that satisfy
the precondition of at least one~$e \in E'$.
One can straightforwardly verify that~$(\MMM \otimes \EEE, (w_0,e_0))$ is
bisimilar to the submodel~$(\MMM',w_0)$ of~$(\MMM,w_0)$ induced by~$W''$.
Moreover,~$\MMM'$ can be computed in polynomial time.

Using this property, we can construct a recursive algorithm
to decide whether~$\MMM,w \models \varphi$.
We consider several cases.
In the case where~$\varphi = p$ for some~$p \in \PPP$,
the problem can easily be solved in polynomial time,
by simply checking whether~$w \in V(p)$.
In the case where~$\varphi = \neg \varphi_1$,
we can recursively call the
algorithm to decide whether~$\MMM,w \models \varphi_1$,
and return the opposite answer.
Similarly, for~$\varphi = \varphi_1 \wedge \varphi_2$,
we can straightforwardly decide whether~$\MMM,w \models \varphi$
by first recursively determining
whether~$\MMM,w \models \varphi_1$
and whether~$\MMM,w \models \varphi_2$.
In the case where~$\varphi = K_{a} \varphi_1$, we firstly
recursively determine whether~$\MMM,w' \models \varphi_1$
for each~$w' \in W$ that is $a$-accessible from~$w$.
This information immediately determines
whether~$\MMM,w \models K_{a} \varphi_1$.

Finally, consider the case where~$\varphi = [\EEE,e] \varphi_1$.
In this case, we firstly recursively decide
whether~$\MMM,w \models \ccfont{pre}(e)$.
If this is not the case, then trivially,~$\MMM,w \models \varphi$.
Otherwise, we construct the submodel~$\MMM'$ of~$\MMM$
that is bisimilar to~$\MMM \otimes \EEE$.
This can be done as described above.
In order to do this, we need to decide which states in~$W'$
satisfy the precondition of some~$e' \in E'$, where~$W' \subseteq W$
and~$E' \subseteq E$ are defined as explained above.
This can be done by recursive calls of the algorithm.
Having determined~$W'$, and having constructed~$\MMM'$,
we can now answer the question whether~$\MMM,w \models [\EEE,e]\varphi_1$
by using only~$\MMM'$,~$w$ and~$\varphi_1$.
We know that~$w$ is a world
in~$\MMM'$, since~$\MMM,w \models \ccfont{pre}(e)$.
Since~$\MMM'$ is bisimilar to~$\MMM \otimes \EEE$,
it holds that~$\MMM \otimes \EEE, (w,e) \models \varphi_1$
if and only if~$\MMM',w \models \varphi_1$.
Therefore, by recursively calling the algorithm to decide
whether~$\MMM',w \models \varphi_1$, we can decide
whether~$\MMM,w \models [\EEE,e] \varphi_1$.

It is straightforward to verify that this recursive algorithm
correctly decides whether~$\MMM,w \models \varphi$.
However, naively executing this recursive algorithm will, in
the worst case result in an exponential running time.
This is because for the case for~$\varphi = [\EEE,e] \varphi_1$,
the algorithm makes multiple (say~$b \geq 2$) recursive calls for~$\ccfont{pre}(e)$,
and~$\ccfont{pre}(e)$ could contain subformulas
of the form~$[\EEE',e'] \varphi'$---which in turn triggers multiple recursive
calls for~$\ccfont{pre}(e')$ for each of the~$b$ branches in the recursion
tree, and so forth.
As the number of these iterations can grow linearly with the input size
(say~$f(n)$),
the recursion tree can be of exponential size
(namely, of size~$\geq 2^{f(n)}$).
We describe how to modify the algorithm to run in
polynomial time, using the technique of memoization.
Whenever a recursive call is made to decide
whether~$\NNN,u \models \psi$, for some
submodel~$\NNN$ of~$\MMM$, some world~$w$
in~$\NNN$, and some subformula~$\psi$ of~$\varphi$,
the result of this recursive call is stored in a lookup table.
Moreover, before making a recursive call to decide
whether~$\NNN,u \models \psi$, the lookup table is
consulted, and if an answer is stored, the algorithm uses
this answer instead of executing the recursive call.

The number of submodels~$\NNN$ of~$\MMM$
that need to be considered in the execution of the modified
algorithm is upper bounded by the number of occurrences
of update operators~$[\EEE,e]$ in the formula~$\varphi$
that is given as input to the problem.
Therefore, the size of the lookup table is polynomial
in the input size.
Moreover, computing the answer for any entry in the lookup
table can be done in polynomial time (using the answers for
other entries in the lookup table).
Therefore, the modified algorithm decides
whether~$\MMM,w \models \varphi$ in polynomial time.
\qed\end{proof}

%%%
%%% HARDNESS RESULTS FOR ONE AGENT
%%%
\subsection{Hardness results for one agent}

Next, we consider the restriction where we have a single agent
and single-pointed models, but where postconditions
are allowed in the event models.%
\footnote{
The reader might wonder for what type of situations the DEL setting
of Theorem~\ref{thm:delta2} (a single agent, single-pointed S5 models,
and postconditions) could be useful.
This restricted setting is relevant, for example, when reasoning about
the epistemic state of a single agent in the face of undercainty over
changes in the world (made by nature).
A simple example of a situation where such reasoning plays a role
is in the analysis of single-player memory games where
the state of (parts of) the game board can be changed randomly
by the rules of the game.}
In this case, the problem is \DeltaP{2}-hard.
{
This hardness result is interesting because it helps identify the
boundaries of the tractable fragment of Proposition~\ref{prop:ptime}.
The result of Theorem~\ref{thm:delta2} shows that adding the single element
of postconditions to this tractable fragment leads to computational
hardness.}

\begin{theorem}
\label{thm:delta2}
The model checking problem for DEL with S5 models
restricted to instances with a single agent
and only single-pointed models,
but where event models can contain postconditions,
is \DeltaP{2}-hard.
\end{theorem}
\begin{proof}
To show \DeltaP{2}-hardness, we give a polynomial-time reduction
from the problem of deciding whether the lexicographically maximal
assignment that satisfies
a given propositional formula~$\varphi$
over variables~$x_1,\dotsc,x_n$
sets the variable~$x_n$ to true.
Let~$\varphi$ be an instance of this problem, with
variables~$x_1,\dotsc,x_n$.
We construct a single-pointed epistemic model~$(\MMM,w_0)$
with a single agent~$a$
and a DEL-formula~$\chi$ whose updates consist of
single-pointed event models (that contain postconditions),
such that~$\MMM,w_0 \models \chi$ if and only
if~$x_n$ is true in the lexicographically maximal assignment
that satisfies~$\varphi$.

In addition to the propositional variables~$x_1,\dotsc,x_n$, we
introduce a variable~$z$.
Then, we construct the model~$(\MMM,w_0)$
as depicted in Figure~\ref{fig:model1}.

\begin{figure}[htp!]
\begin{center}
\begin{tikzpicture}
  \tikzstyle{dnode}=[inner sep=1pt,outer sep=1pt,draw,circle,minimum width=9pt]
  \tikzstyle{nnode}=[inner sep=1pt,outer sep=1pt,circle,minimum width=9pt]
  \tikzstyle{label-edge}=[midway,fill=white, inner sep=1pt]
  % NODES
  \node[dnode, label=above:{\scriptsize $z, \neg x_1, \dotsc, \neg x_n$}] (w0) at (0,0) {};
  \fill (w0) circle [radius=2pt];
  \node[nnode, label=above:{\scriptsize $\neg z, \neg x_1, \dotsc, \neg x_n$}] (w1) at (3,0) {};
  \fill (w1) circle [radius=2pt];
  % EDGES
  \draw[-] (w0) -- (w1) node[label-edge] {\scriptsize $a$};
\end{tikzpicture}
\end{center}
\caption{The epistemic model~$(\MMM,w_0)$,
used in the proof of Theorem~\ref{thm:delta2}.}
\label{fig:model1}
\end{figure}
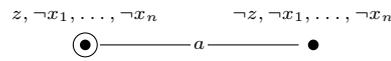

Then, for each~$1 \leq i \leq n$, we introduce the single-pointed
event model~$(\EEE_i,e_i)$ as depicted in Figure~\ref{fig:update1}.
Intuitively, these updates will serve to generate, for each possible
truth assignment~$\alpha$ to the variables~$x_1,\dotsc,x_n$,
a world that agrees with~$\alpha$ (and that sets~$z$ to false),
in addition to the designated world (where~$z$ is true).
We will denote the model resulting from updating~$(\MMM,w_0)$
subsequently
with the updates~$(\EEE_1,e_1),\dotsc,(\EEE_n,e_n)$ by~$(\MMM',w')$.

\begin{figure}[htp!]
\begin{center}
\begin{tikzpicture}
  \tikzstyle{dnode}=[inner sep=1pt,outer sep=1pt,draw,circle,minimum width=9pt]
  \tikzstyle{nnode}=[inner sep=1pt,outer sep=1pt,circle,minimum width=9pt]
  \tikzstyle{dnode'}=[inner sep=1pt,outer sep=1pt,draw,rectangle,minimum width=9pt,minimum height=9pt]
  \tikzstyle{nnode'}=[inner sep=1pt,outer sep=1pt,rectangle,minimum width=9pt,minimum height=9pt]  \tikzstyle{label-edge}=[midway,fill=white, inner sep=1pt]
  % NODES
  \node[dnode', label=above:{\scriptsize $\langle z, \top \rangle$}] (w0) at (3,1.5) {};
  \fill {(w0)+(-0.07,-0.07)} rectangle (3.07,1.57);
  \node[nnode', label=below:{\scriptsize $\langle \neg z, x_i \rangle$}] (w1) at (2,0) {};
  \fill {(w1)+(-0.07,-0.07)} rectangle (2.07,0.07);
  \node[nnode', label=below:{\scriptsize $\langle \neg z, \top \rangle$}] (w2) at (4,0) {};
  \fill {(w2)+(-0.07,-0.07)} rectangle (4.07,0.07);
  % EDGES
  \draw[-] (w0) -- (w1) node[label-edge] {\scriptsize $a$};
  \draw[-] (w1) -- (w2) node[label-edge] {\scriptsize $a$};
  \draw[-] (w0) -- (w2) node[label-edge] {\scriptsize $a$};
\end{tikzpicture}
\end{center}
\caption{The event model~$(\EEE_i,e_i)$,
used in the proof of Theorem~\ref{thm:delta2}.}
\label{fig:update1}
\end{figure}

Next, for each~$1 \leq i \leq n$, we introduce the single-pointed event
model~$(\EEE'_i,e'_i)$ as depicted in Figure~\ref{fig:update2}.
Intuitively, we will use the event models~$(\EEE'_i,e'_i)$ to obtain
(many copies of)
the lexicographically maximal assignment that
satisfies~$\varphi$.
Applying the update~$(\EEE'_i,e'_i)$ to~$(\MMM',w')$ will
set the variable~$x_i$ to true in all worlds (that satisfy~$\neg z$)
if there is an assignment
(among the remaining assignments)
that satisfies~$\varphi$ and
that sets~$x_i$ to true,
and it will set the variable~$x_i$ to false in all worlds
(that satisfy~$\neg z$) otherwise.
Then, after applying the
updates~$(\EEE'_1,e'_1),\dotsc,(\EEE'_n,e'_n)$
to~$(\MMM',w')$,
all worlds in the resulting model will have the same
valuation---namely, a valuation
that agrees with the lexicographically maximal assignment
that satisfies~$\varphi$.
In particular,
the variable~$x_n$ is true in this valuation if and
only if~$x_n$ is true in the lexicographically maximal
assignment that satisfies~$\varphi$.

\begin{figure}[htp!]
\begin{center}
\begin{tikzpicture}
  \tikzstyle{dnode}=[inner sep=1pt,outer sep=1pt,draw,circle,minimum width=9pt]
  \tikzstyle{nnode}=[inner sep=1pt,outer sep=1pt,circle,minimum width=9pt]
  \tikzstyle{dnode'}=[inner sep=1pt,outer sep=1pt,draw,rectangle,minimum width=9pt,minimum height=9pt]
  \tikzstyle{nnode'}=[inner sep=1pt,outer sep=1pt,rectangle,minimum width=9pt,minimum height=9pt] 
  \tikzstyle{label-edge}=[midway,fill=white, inner sep=1pt]  \tikzstyle{label-edge}=[midway,fill=white, inner sep=1pt]
  % NODES
  \node[dnode', label=above:{\scriptsize $\langle z, \top \rangle$}] (w0) at (3.75,1.5) {};
  \fill {(w0)+(-0.07,-0.07)} rectangle (3.82,1.57);
  \node[nnode', label=below:{\scriptsize $\langle \neg z \wedge \hat{K}_a (x_i \wedge \varphi), x_i \rangle$}] (w1) at (2,0) {};
  \fill {(w1)+(-0.07,-0.07)} rectangle (2.07,0.07);
  \node[nnode', label=below:{\scriptsize $\langle \neg z \wedge \neg \hat{K}_a (x_i \wedge \varphi), \neg x_i \rangle$}] (w2) at (5.5,0) {};
  \fill {(w2)+(-0.07,-0.07)} rectangle (5.57,0.07);
  % EDGES
  \draw[-] (w0) -- (w1) node[label-edge] {\scriptsize $a$};
  \draw[-] (w1) -- (w2) node[label-edge] {\scriptsize $a$};
  \draw[-] (w0) -- (w2) node[label-edge] {\scriptsize $a$};
\end{tikzpicture}
\end{center}
\caption{The event model~$(\EEE'_i,e'_i)$,
used in the proof of Theorem~\ref{thm:delta2}.}
\label{fig:update2}
\end{figure}

%We then define the formula~$\chi$ as follows:
%\[ \chi = [\EEE_1,e_1]\dotsc[\EEE_n,e_n]
%[\EEE'_1,e'_1]\dotsc[\EEE'_n,e'_n] \hat{K}_a x_n. \]
%
We then let~$\chi = [\EEE_1,e_1]\dotsc[\EEE_n,e_n]
[\EEE'_1,e'_1]\dotsc[\EEE'_n,e'_n] \hat{K}_a x_n$.
We now formally show that the lexicographically maximal
assignment~$\alpha_0$ that satisfies~$\varphi$
sets~$x_n$ to true if and only if~$\MMM,w_0 \models \chi$.
In order to do so, we will prove the following claim.
The model~$(\MMM'',w'') = (\MMM,w_0) \otimes (\EEE_1,e_1) \otimes
\dotsm \otimes (\EEE_n,e_n) \otimes (\EEE'_1,e'_1) \otimes \dotsm \otimes (\EEE'_n,e'_n)$
consists of a world~$w''$ that sets~$z$ to true and all other variables
to false, and of worlds that set~$z$ to false and that agree with~$\alpha_0$ on
the variables~$x_1,\dotsc,x_n$.
Firstly, it is straightforward to verify that~$(\MMM',w') = (\MMM,w_0) \otimes
(\EEE_1,e_1) \otimes \dotsm \otimes (\EEE_n,e_n)$ consists of the
world~$w'$ and exactly one world corresponding to each truth
assignment~$\alpha$ to the variables~$x_1,\dotsc,x_n$.

Then, applying the update~$(\EEE'_1,e'_1)$ to~$(\MMM',w')$ has two
possible outcomes: either~(1) if there exists a model of~$\varphi$ that
sets~$x_1$ to true, then in all worlds (that set~$z$ to false)
the variable~$x_1$ will be set to true;
or~(2) if there exists no model of~$\varphi$ that sets~$x_1$ to true,
then in all worlds (that set~$z$ to false) the variable~$x_1$ will be
set to false.
For each~$1 < i \leq n$, subsequently applying the update~$(\EEE'_i,e'_i)$
has an entirely similar effect.
By a straightforward inductive argument, it then follows that
all the worlds in~$(\MMM'',w'')$ that set~$z$ to false agree
with the lexicographically maximal model of~$\varphi$.

Therefore,~$\MMM,w_0 \models \chi$
%~$\MMM,w_0 \models [\EEE_1,e_1]\dotsc[\EEE'_n,e'_n] \hat{K}_a x_n$
if and only
if~$x_n$ is true in the lexicographically maximal model of~$\varphi$,
and we can conclude that the problem is \DeltaP{2}-hard.
\qed\end{proof}

%%
%% NOTE: it would be interesting to see what happens in the case
%% where there are only a constant number of propositions.
%% Would it be possible to show a poly-time algorithm?
%%
When we allow multi-pointed models, the problem turns
out to be \PSPACE{}-hard, even when restricted to a single agent
(Theorem~\ref{thm:pspace-hardness2}).
{
This hardness result adds to our understanding of the boundaries
of the algorithmically tractable fragment of Proposition~\ref{prop:ptime}.
Whereas Theorem~\ref{thm:delta2} showed that adding postconditions
leads to intractability, the following result shows that adding
multi-pointedness to the models instead (and having no further
additions) also leads to intractability.
In other words, the following result indicates that leaving the fragment
of Proposition~\ref{prop:ptime} by a different route also requires
giving up polynomial-time algorithms for model checking.
}

{
In the literature, \PSPACE{}-hardness results have been shown
for a setting that is similar to the one used in
Theorem~\ref{thm:pspace-hardness2}---i.e.,~%
\cite[Proposition~7.2]{VanEijckSchwarzentruber14}
and~\cite[Theorem~2]{AucherSchwarzentruber13}.
The difference is that these proofs in the literature depend
on particular features of the DEL setting---%
the result of \cite[Proposition~7.2]{VanEijckSchwarzentruber14}
depends on there being two agents,
and the result of \cite[Theorem~2]{AucherSchwarzentruber13}
depends on relations not being serial---%
whereas the result of
Theorem~\ref{thm:pspace-hardness2} holds also for the case
with both a single agent and S5 relations.
The proofs from the literature depends crucially on
there being two agents or non-serial relations---%
they are used to encode the quantifiers in a quantified Boolean formula.
The main technical hurdle that needs to be overcome to establish
Theorem~\ref{thm:pspace-hardness2} is to encode the quantification
of a quantified Boolean formula in DEL using a single agent
and using S5 relations.
We do so by starting with an S5 epistemic model~$\mathcal{M}$
that includes worlds that set different propositional
variables~$x_1,\dotsc,x_n$ to true,
and using multi-pointed event models to quantify over different
possibilities of deleting worlds from~$\mathcal{M}$.}

\begin{theorem}
\label{thm:pspace-hardness2}
The model checking problem for DEL with S5 models
restricted to instances with a single agent
and no postconditions in the event models,
but where models can be multi-pointed,
is \PSPACE{}-hard.
\end{theorem}
\begin{proof}
In order to show \PSPACE{}-hardness, we give a polynomial-time reduction
from the problem of deciding whether a quantified Boolean formula is true.
Let~$\varphi = \exists x_1 \forall x_2 \dotsc \exists x_{n-1} \forall x_{n}. \psi$
be a quantified Boolean formula, where~$\psi$ is quantifier-free
(we assume without loss of generality that~$n$ is even).
We construct a single-pointed epistemic model~$(\MMM,w_0)$
with one agent~$a$
and a DEL-formula~$\chi$ (containing updates with multi-pointed event
models) such that~$\MMM,w_0 \models \chi$ if and only
if~$\varphi$ is true.

The first main idea behind this reduction is that we represent truth assignments
to the propositional variables~$x_1,\dotsc,x_n$ with connected groups
of worlds.
Let~$\alpha$ be a truth assignment to the variables~$x_1,\dotsc,x_n$,
and let~$x_{i_1},\dotsc,x_{i_\ell}$ be the variables that~$\alpha$ sets
to true.
We then represent~$\alpha$ by means of a group of
worlds~$w_0,w_1,\dotsc,w_{\ell}$, where the world~$w_0$ makes no
propositional variable true,
and for each~$1 \leq j \leq \ell$, world~$w_j$ makes exactly
one propositional variable true (namely,~$x_{i_j}$).
These worlds~$w_0,w_1,\dotsc,w_{\ell}$ are fully connected.
This collection of worlds~$w_0,w_1,\dotsc,w_{\ell}$ is what we call
the \emph{group of worlds corresponding to~$\alpha$}.
Moreover, the designated state is~$w_0$.
Consider the truth assignment~$\alpha = \SBs
x_1 \mapsto 1, x_2 \mapsto 1, x_3 \mapsto 0, x_4 \mapsto 1 \SEs$,
for example.
In Figure~\ref{fig:group-example0} we show the group of
worlds that we use to represent this truth assignment~$\alpha$.
We let the model~$\MMM$ be the group of worlds
corresponding to the truth assignment~$\alpha_0$ that assigns
all variables~$x_1,\dotsc,x_n$ to true.

\begin{figure}[htp!]
\begin{center}
\begin{tikzpicture}
  \tikzstyle{dnode}=[inner sep=1pt,outer sep=1pt,draw,circle,minimum width=9pt]
  \tikzstyle{nnode}=[inner sep=1pt,outer sep=1pt,circle,minimum width=9pt]
  \tikzstyle{label-edge}=[midway,fill=white, inner sep=1pt]
  % NODES
  \node[dnode] (w0) at (0,2) {};
  \fill (w0) circle [radius=2pt];
  \node[nnode, label=above:{\scriptsize $x_1$}] (w1) at (3,2) {};
  \fill (w1) circle [radius=2pt];
  \node[nnode, label=below:{\scriptsize $x_2$}] (w2) at (0,0) {};
  \fill (w2) circle [radius=2pt];
  \node[nnode, label=below:{\scriptsize $x_4$}] (w3) at (3,0) {};
  \fill (w3) circle [radius=2pt];
  % EDGES
  \draw[-] (w0) -- (w1) node[label-edge] {\scriptsize $a$};
  \draw[-] (w0) -- (w2) node[label-edge] {\scriptsize $a$};
  \draw[-] (w0) -- (w3) node[label-edge, near start] {\scriptsize $a$};
  \draw[-] (w1) -- (w2) node[label-edge, near start] {\scriptsize $a$};
  \draw[-] (w1) -- (w3) node[label-edge] {\scriptsize $a$};
  \draw[-] (w2) -- (w3) node[label-edge] {\scriptsize $a$};
\end{tikzpicture}
\end{center}
\caption{The group of worlds that we use to represent
the truth assignment~$\alpha = \SBs
x_1 \mapsto 1, x_2 \mapsto 1, x_3 \mapsto 0, x_4 \mapsto 1 \SEs$,
in the proof of Theorem~\ref{thm:pspace-hardness2}.}
\label{fig:group-example0}
\end{figure}
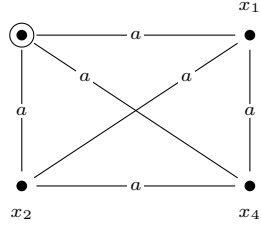

The next main idea is that we represent existential and universal
quantification of the propositional variables using the dynamic
operators~$\langle \EEE,E \rangle$ and~$[\EEE,E]$,
respectively.
%We do this in the following way.
For each propositional variable~$x_i$ in the quantified
Boolean formula,
we introduce the multi-pointed event model~$(\EEE_i,E_i)$
as depicted in Figure~\ref{fig:update-example0}.
We use the event models~$(\EEE_1,E_1),\dotsc,(\EEE_n,E_n)$,
to create (disconnected) groups of worlds (that all have a designated world) that correspond to each possible
truth assignment~$\alpha$ to the variables~$x_1,\dotsc,x_n$.

\begin{figure}[htp!]
\begin{center}
\begin{tikzpicture}
  \tikzstyle{dnode}=[inner sep=1pt,outer sep=1pt,draw,circle,minimum width=9pt]
  \tikzstyle{nnode}=[inner sep=1pt,outer sep=1pt,circle,minimum width=9pt]
  \tikzstyle{dnode'}=[inner sep=1pt,outer sep=1pt,draw,rectangle,minimum width=9pt,minimum height=9pt]
  \tikzstyle{nnode'}=[inner sep=1pt,outer sep=1pt,rectangle,minimum width=9pt,minimum height=9pt]
  \tikzstyle{label-edge}=[midway,fill=white, inner sep=1pt]
  % NODES
  \node[dnode', label=above:{\scriptsize $\langle \top, \top \rangle$}] (w0) at (0,0) {};
  \fill {(w0)+(-0.07,-0.07)} rectangle (0.07,0.07);
  \node[dnode', label=above:{\scriptsize $\langle \neg x_i, \top \rangle$}] (w1) at (2,0) {};
  \fill {(w1)+(-0.07,-0.07)} rectangle (2.07,0.07);
\end{tikzpicture}
\end{center}
\caption{The multi-pointed event model~$(\EEE_i,E_i)$
corresponding to variable~$x_i$,
used in the proof of Theorem~\ref{thm:pspace-hardness2}.}
\label{fig:update-example0}
\end{figure}
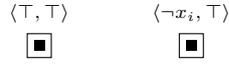

%Intuitively,
%using the event models~$(\EEE_1,E_1),\dotsc,(\EEE_n,E_n)$,
%we can create groups of worlds that correspond to each possible
%truth assignment~$\alpha$ to the variables~$x_1,\dotsc,x_n$.
%

Using the alternation of diamond dynamic operators
and box dynamic operators, we can
simulate existential and universal quantification of variables 
%the behavior of the existential and universal quantifiers 
in the formula~$\varphi$.
%We do this by
%using the updates~$(\EEE_1,E_1),\dotsc,(\EEE_n,E_n)$,
%to make (disconnected)
%copies of the worlds, each of which has a designated world.
%Each of these copies corresponds to one possible truth
%assignment~$\alpha$ to~$x_1,\dotsc,x_n$.
%
%By constructing these copies of worlds corresponding to the different
%truth assignments in this way, we can simulate existential and universal
%quantification of variables in~$\varphi$.
We simulate an existentially quantified variable~$\exists x_i$
by the dynamic operator~$\langle \EEE_i,E_i \rangle$---%
a formula of the form~$\langle \EEE_i,E_i \rangle \phi$ is true
if and only if~$\langle \EEE_i,e_i \rangle \phi$ is true for some~$e_i \in E_i$.
Similarly, we simulate a universally quantified variable~$\forall x_i$
by the dynamic operator~$[\EEE_i,E_i]$---%
a formula of the form~$[\EEE_i,E_i] \phi$ is true
if and only~$[\EEE_i,e_i] \phi$ is true for all~$e_i \in E_i$.

%We construct the formula~$\chi$ as follows:
%\[ \chi = \langle \EEE_1,E_1 \rangle [\EEE_2,E_2] \dotsc
%\langle \EEE_{n-1}, E_{n-1} \rangle [\EEE_n,E_n] \chi' \]
Concretely, we let~%
$\chi = \langle \EEE_1,E_1 \rangle [\EEE_2,E_2] \dotsc
\langle \EEE_{n-1}, E_{n-1} \rangle [\EEE_n,E_n] \chi'$,
where~$\chi'$ is the formula obtained from~$\psi$
by replacing each occurrence of a propositional variable~$x_i$
by the formula~$\hat{K}_a x_i$.

%We claim that~$\MMM,w_0 \models \chi$ if and only if~$\varphi$ is a true
%quantified Boolean formula.
%\textcolor{red}{A detailed proof of this can be found in the appendix.}

We show that~$\varphi$ is a true
quantified Boolean formula if and only if~$\MMM,w_0 \models \chi$.
In order to do so, we prove the following statement, relating
truth assignments~$\alpha$ to the variables~$x_1,\dotsc,x_n$ to
groups of worlds containing a designated world.
The statement that we will prove inductively
for all~$1 \leq i \leq n+1$ is the following.

\medskip
\noindent \textit{Statement:}
Let~$\alpha$ be any truth assignment to the variables~$x_1,\dotsc,x_n$
that sets all variables~$x_i,\dotsc,x_n$ to true,
and let~$\alpha'$ be the restriction of~$\alpha$ to the
variables~$x_1,\dotsc,x_{i-1}$.
Moreover, let~$M$ be a group of worlds that corresponds to the
truth assignment~$\alpha$, containing a designated world~$w$.
Then~$Q_i x_i \dotsc \exists x_{n-1} \forall x_n. \psi$ is true
under~$\alpha'$ if and only if:
\begin{itemize}
  \item $w$ makes $\langle \EEE_i,E_i \rangle \dotsc
%    \langle \EEE_{n-1}, E_{n-1} \rangle
    [\EEE_n,E_n] \chi'$ true, if~$i$ is odd; and
  \item $w$ makes $[\EEE_i,E_i] \dotsc
%    \langle \EEE_{n-1}, E_{n-1} \rangle
    [\EEE_n,E_n] \chi'$ true, if~$i$ is even.
\end{itemize}
\smallskip

The statement for~$i=1$ implies that~$\MMM,w_0 \models \chi$
if and only if~$\varphi$ is a true quantified Boolean formula.
We show that the statement holds for~$i=1$
by showing that the statement holds for all~$1 \leq i \leq n+1$.
We begin by showing that the statement holds for~$i = n+1$.
In this case, we know that~$\alpha = \alpha'$ is a truth assignment to the
variables~$x_1,\dotsc,x_n$.
Moreover, by construction of~$\chi'$ we know
that~$w$ makes~$\chi'$ true if and only if~$\alpha$ satisfies~$\psi$.
Therefore, the statement holds.

Next, we let~$1 \leq i \leq n$ be arbitrary, and we assume that the
statement holds for~$i+1$.
%that is, the statement holds for every combination of a truth
%assignment~$\alpha$, a set~$M$ of worlds and a world~$w$
%that meets the requirements.
%
We now distinguish two cases: either~(1)~$Q_i = \exists$, i.e., the $i$-th
quantifier of~$\varphi$ is existential,
or~(2)~$Q_i = \forall$, i.e., the $i$-th quantifier of~$\varphi$
is universal.

First, consider case~(1).
Suppose that~$\exists x_i \dotsc \exists x_{n-1} \forall x_n. \psi$ is true
under~$\alpha'$.
Then there exists a truth assignment~$\alpha''$ to the
variables~$x_1,\dotsc,x_i$ that agrees with~$\alpha'$ on the
variables~$x_1,\dotsc,x_{i-1}$ and for which~$\forall x_{i+1} \dotsc
\exists x_{n-1} \forall x_n. \psi$ is true under~$\alpha''$.
%\textcolor{red}{[quantifier scope (e.g., for~$M$) is a bit implicit/sloppy]}
Therefore, there exists some event~$e \in E_i$ such
that the group~$M' = \SB (v,e) \SM v \in M \mtext{ and } M,v \models \ccfont{pre}(e) \SE$
of worlds and
the world~$w' = (w,e)$, together with the assignment~$\alpha'''$
that agrees with~$\alpha''$ on the variables~$x_1,\dotsc,x_i$
and that sets the variables~$x_{i+1},\dotsc,x_n$ to true,
satisfy the requirements for the statement for~$i+1$.
Then, by the induction hypothesis we know
that~$w'$ makes~$[\EEE_{i+1},E_{i+1}] \dotsc \allowbreak
\langle \EEE_{n-1}, E_{n-1} \rangle [\EEE_n,E_n] \chi'$ true.
Therefore, we can conclude that~$w$
makes~$\langle \EEE_{i},E_{i} \rangle \allowbreak \dotsc
\langle \EEE_{n-1}, E_{n-1} \rangle [\EEE_n,E_n] \chi'$ true.

Conversely, suppose that~$w$
makes~$\langle \EEE_{i},E_{i} \rangle \dotsc
\langle \EEE_{n-1}, E_{n-1} \rangle [\EEE_n,E_n] \chi'$ true.
This can only be the case if there is some event~$e \in E_i$
such that the set~$M'$ and~$w'$ (defined as above)
correspond to a truth assignment~$\alpha'''$ (also defined as above).
Then, by the induction hypothesis, we know
that~$\forall x_{i+1} \dotsc \exists x_{n-1} \forall x_n. \psi$
is true under~$\alpha''$ (obtained from~$\alpha'''$ as above).
Therefore, since~$\alpha''$ extends~$\alpha'$,
we can conclude that~$\exists x_i \dotsc \exists x_{n-1}
\forall x_n. \psi$ is true under~$\alpha'$.

The argument for case~(2) is entirely analogous (yet dual).
We omit a detailed treatment of this case.
This concludes the inductive proof of the statement for
all~$1 \leq i \leq n+1$, and thus concludes our proof
that~$\MMM,w_0 \models \chi$ if and only
if~$\varphi$ is true.
\qed\end{proof}

%%%
%%% HARDNESS RESULTS FOR TWO AGENTS
%%%
\subsection{Hardness results for two agents}

Next, we show that when we consider the case of two agents,
the model checking problem for DEL is \PSPACE{}-hard,
even when we only allow single-pointed models
without postconditions (Theorem~\ref{thm:pspace-hardness}).
{
This hardness result adds yet another piece of understanding
of the boundaries
of the algorithmically tractable fragment of Proposition~\ref{prop:ptime}.
Namely, it shows that leaving the algorithmically tractable fragment
of Proposition~\ref{prop:ptime} by another one of the possible different
routes---that is, by adding a second agent only---%
leads to computational hardness.
In fact, this is the third of the three most obvious ways of extending
the fragment of Proposition~\ref{prop:ptime}.
Therefore, Theorem~\ref{thm:pspace-hardness}---together with 
Theorems~\ref{thm:delta2} and~\ref{thm:pspace-hardness2}---%
indicates that all individual restrictions in the fragment of
Proposition~\ref{prop:ptime} are necessary to obtain polynomial-time
solvability.

Additionally, the result of Theorem~\ref{thm:pspace-hardness} shows
that the inherent hardness in the model checking problem for DEL
with two agents holds even when we restrict
the setting to include only three propositional variables~$x_1,x_2,x_3$---%
in addition to having only two agents and single-pointed S5 models only.

Similarly to the case of Theorem~\ref{thm:pspace-hardness2},
the result of Theorem~\ref{thm:pspace-hardness} differs from similar results
from the literature---i.e.,~%
\cite[Proposition~7.2]{VanEijckSchwarzentruber14}
and~\cite[Theorem~2]{AucherSchwarzentruber13}---%
in that it considers different restrictions than these results.
In particular, the former result \cite[Proposition~7.2]{VanEijckSchwarzentruber14}
requires the use of multi-pointed models,
and the latter result \cite[Theorem~2]{AucherSchwarzentruber13} requires
non-serial relations in the models
(and these proofs crucially depend on these features).
In the literature, there are also \PSPACE{} hardness results for the case
of single-pointed S5 models \cite{VandePol15,VandePolVanRooijSzymanik18}.
However, these results hold for the case where an unbounded number of agents
are used---and the proofs given in the literature
crucially depend on the number of agents not being bounded.
The result of Theorem~\ref{thm:pspace-hardness} holds for the case
with only two agents and where only single-pointed S5 models are allowed.

The main technical hurdle that needs to be overcome to establish
Theorem~\ref{thm:pspace-hardness} is to encode the differently quantified
variables of a quantified Boolean formula using single-pointed S5 models
and using only two agents and three propositional variables.
We do so---roughly---by (i)~representing propositional variables
by alternating chains of worlds where the end of the chain is marked by
a designated propositional variable~$z_0$,%
\footnote{ This technique was previously used by Van de Pol, Van Rooij, and
Szymanik in a hardness proof for DEL model checking
\cite[Proposition~3]{VandePolVanRooijSzymanik18}.}
(ii)~representing quantification over different variables in the quantified
Boolean formula using S5 event models where different alternating chains
of relations lead to different copies of the original model that represent
different truth assignments to the variables of the quantified Boolean formula,
and (iii)~by introducing DEL formulas that interact appropriately
with the gadgets of~(i) and~(ii), making use of only two
auxiliary variables~$z_1$ and~$z_2$.
{
The main challenge in~(i)--(iii) is in the details of
the intricate construction of
the DEL formulas and in the detailed argument that they interact
in exactly the right way with the event models---this is why the
proof of Theorem~\ref{thm:pspace-hardness} is of considerable length.
}
}

\begin{theorem}
\label{thm:pspace-hardness}
The model checking problem for DEL is \PSPACE{}-hard,
even when restricted to the case where the question is
whether~$\MMM,w_0 \models [\EEE_1,e_1]\dotsc[\EEE_n,e_n] \chi$,
where:
\begin{itemize}
  \item the model~$(\MMM,w_0)$ is a single-pointed S5 model;
  \item all the~$(\EEE_i,e_i)$ are single-pointed S5 event models
    without postconditions;
  \item $\chi$ is an epistemic formula without update modalities
    that contains
    (multiple occurrences of)
    only three propositional variables; and
  \item there are only two agents.
\end{itemize}
\end{theorem}
\begin{proof}
In order to show \PSPACE{}-hardness, we give a polynomial-time reduction
from the problem of deciding whether a quantified Boolean formula is true.
Let~$\varphi = \exists x_1 \forall x_2 \dotsc \exists x_{n-1} \forall x_{n}. \psi$
be a quantified Boolean formula, where~$\psi$ is quantifier-free.
%(We may assume without loss of generality that~$\psi$ does not contain
%any occurrences of~$x_n$. If this were not the case, we could simply add
%dummy variables~$x_{n+1},x_{n+2}$ to ensure this.
%We will use this assumption later in the proof for technical reasons.)
We construct an epistemic model~$(\MMM,w_0)$ with two agents~$a,b$
and a DEL-formula~$\xi$ such that~$\MMM,w_0 \models \xi$ if and only
if~$\varphi$ is true.

%The general idea behind the following proof is similar
%to that of a previous \PSPACE{}-hardness proof for DEL model
%checking in the literature
%\cite{VandePol15,VandePolVanRooijSzymanik18}.
%In particular, the way in which we represent propositions
%and truth assignments is similar.

The first main idea behind the reduction is that we use two propositional
variables, say~$z_0$ and~$z_1$, to represent an arbitrary number of propositions,
by creating chains of worlds that represent these propositions.
Let~$x_1,\dotsc,x_n$ be the propositions that we want to represent.
Then we represent a proposition~$x_j$ by a chain of worlds of
length~$j+1$
that are connected alternatingly by $b$-relations and $a$ relations.
In this chain, the last world is the only world that makes~$z_0$ true.
Moreover, the first world in the chain is the only world
that makes~$z_1$ true.
An example of such a chain that we use to represent proposition~$x_3$ 
can be found in Figure~\ref{fig:chain-example}.

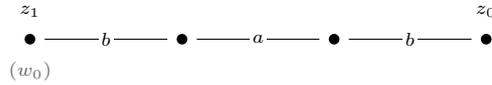
\begin{figure}[htp!]
\begin{center}
\begin{tikzpicture}
  \tikzstyle{dnode}=[inner sep=1pt,outer sep=1pt,draw,circle,minimum width=9pt]
  \tikzstyle{nnode}=[inner sep=1pt,outer sep=1pt,circle,minimum width=9pt]
  \tikzstyle{label-edge}=[midway,fill=white, inner sep=1pt]
  % NODES
  \node[nnode, label=below:{\scriptsize \textcolor{black!60}{$(w_0)$}}, label=above:{\scriptsize $z_1$}] (w0) at (0,0) {};
  \fill (w0) circle [radius=2pt];
  \node[nnode] (w1) at (2,0) {};
  \fill (w1) circle [radius=2pt];
  \node[nnode] (w2) at (4,0) {};
  \fill (w2) circle [radius=2pt];
  \node[nnode, label=above:{\scriptsize $z_0$}] (w3) at (6,0) {};
  \fill (w3) circle [radius=2pt];
  % EDGES
  \draw[-] (w0) -- (w1) node[label-edge] {\scriptsize $b$};
  \draw[-] (w1) -- (w2) node[label-edge] {\scriptsize $a$};
  \draw[-] (w2) -- (w3) node[label-edge] {\scriptsize $b$};
\end{tikzpicture}
\end{center}
\caption{The chain of worlds that we use to represent proposition~$x_3$
in the proof of Theorem~\ref{thm:pspace-hardness}.
The first world in the chain is the world~$w_0$, that is depicted
on the left.}
\label{fig:chain-example}
\end{figure}

The next idea that plays an important role in the reduction is that we will
group together (the first worlds) of several such chains to represent
a truth assignment to the propositions~$x_1,\dotsc,x_n$.
Let~$\alpha$ be a truth assignment to the propositions~$x_1,\dotsc,x_n$,
and let~$x_{i_1},\dotsc,x_{i_\ell}$ be the propositions that~$\alpha$
sets to true (for~$1 \leq i_1 < \dotsm < i_{\ell} \leq n$).
Then we represent the truth assignment~$\alpha$ in the following way.
We take the chains corresponding to the
propositions~$x_{i_1},\dotsc,x_{i_\ell}$,
and we connect the first world (the world that is labelled with~$w_0$ in
Figure~\ref{fig:chain-example}) of each two of these chains with
an $a$-relation.
In other words, we connect all these first worlds together in a fully connected
clique of $a$-relations.
Moreover, to this $a$-clique of worlds, we add a designated world
where both~$z_1$ and a third propositional variable~$z_2$ are true.
The collection of all worlds in the chains corresponding to
the propositions~$x_{i_1},\dotsc,x_{i_\ell}$ and this additional
designated world
is what we call the \emph{group of worlds
representing~$\alpha$}.
For the sake of convenience, we call the world where~$z_1$
and~$z_2$ are true the \emph{central world} of the group of worlds.
In Figure~\ref{fig:group-example} we give an example of such a group of
worlds that we use to represent the truth assignment~$\alpha = \SBs
x_1 \mapsto 1, x_2 \mapsto 1, x_3 \mapsto 0, x_4 \mapsto 1 \SEs$.

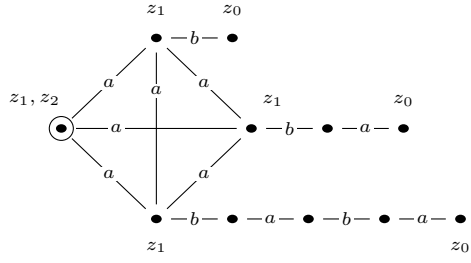
\begin{figure}[htp!]
\begin{center}
\begin{tikzpicture}[yscale=0.8]
  \tikzstyle{dnode}=[inner sep=1pt,outer sep=1pt,draw,circle,minimum width=9pt]
  \tikzstyle{nnode}=[inner sep=1pt,outer sep=1pt,circle,minimum width=9pt]
  \tikzstyle{label-edge}=[midway,fill=white, inner sep=1pt]
  \node[dnode, label={[xshift=-10pt]above:{\scriptsize $z_1,z_2$}}] (z) at (-1.25,-1.5) {};
  \fill (z) circle [radius=2pt];
  % CHAIN 1
  \node[nnode, label=above:{\scriptsize $z_1$}] (u0) at (0,0) {};
  \fill (u0) circle [radius=2pt];
  \node[nnode, label=above:{\scriptsize $z_0$}] (u1) at (1,0) {};
  \fill (u1) circle [radius=2pt];
  \draw[-] (u0) -- (u1) node[label-edge] {\scriptsize $b$};
  % CHAIN 2
  \node[nnode, label={[xshift=8pt]above:{\scriptsize $z_1$}}] (v0) at (1.25,-1.5) {};
  \fill (v0) circle [radius=2pt];
  \node[nnode] (v1) at (2.25,-1.5) {};
  \fill (v1) circle [radius=2pt];
  \node[nnode, label=above:{\scriptsize $z_0$}] (v2) at (3.25,-1.5) {};
  \fill (v2) circle [radius=2pt];
  \draw[-] (v0) -- (v1) node[label-edge] {\scriptsize $b$};
  \draw[-] (v1) -- (v2) node[label-edge] {\scriptsize $a$};
  % CHAIN 4
  \node[nnode, label=below:{\scriptsize $z_1$}] (w0) at (0,-3) {};
  \fill (w0) circle [radius=2pt];
  \node[nnode] (w1) at (1,-3) {};
  \fill (w1) circle [radius=2pt];
  \node[nnode] (w2) at (2,-3) {};
  \fill (w2) circle [radius=2pt];
  \node[nnode] (w3) at (3,-3) {};
  \fill (w3) circle [radius=2pt];
  \node[nnode, label=below:{\scriptsize $z_0$}] (w4) at (4,-3) {};
  \fill (w4) circle [radius=2pt];
  \draw[-] (w0) -- (w1) node[label-edge] {\scriptsize $b$};
  \draw[-] (w1) -- (w2) node[label-edge] {\scriptsize $a$};
  \draw[-] (w2) -- (w3) node[label-edge] {\scriptsize $b$};
  \draw[-] (w3) -- (w4) node[label-edge] {\scriptsize $a$};
  % CLIQUE
  \draw[-] (u0) -- (v0) node[label-edge] {\scriptsize $a$};
  \draw[-] (u0) -- (w0) node[label-edge, near start] {\scriptsize $a$};
  \draw[-] (v0) -- (w0) node[label-edge] {\scriptsize $a$};
  \draw[-] (z) -- (u0) node[label-edge] {\scriptsize $a$};
  \draw[-] (z) -- (v0) node[label-edge, near start] {\scriptsize $a$};
  \draw[-] (z) -- (w0) node[label-edge] {\scriptsize $a$};
\end{tikzpicture}
\end{center}
\caption{The group of worlds that we use to represent the truth
assignment~$\alpha = \SBs
y_1 \mapsto 1, y_2 \mapsto 1, y_3 \mapsto 0, y_4 \mapsto 1 \SEs$
in the proof of Theorem~\ref{thm:pspace-hardness}.}
\label{fig:group-example}
\end{figure}

By using the expressivity of epistemic logic, we can construct formulas
that extract information from these representations of truth assignments.
Intuitively, we can check whether a truth assignment~$\alpha$
sets a proposition~$x_j$ to true by checking whether the group of worlds
representing~$\alpha$ contains a chain of worlds of length exactly~$j+1$.
Formally, we will define a formula~$\chi_j$ for each~$1 \leq j \leq n$,
which is true in the designated world
if and only if the group contains a chain representing proposition~$x_j$.
We describe how to construct the formulas~$\chi_j$.
Firstly, we inductively define formulas~$\chi^{a}_{j}$
and~$\chi^{b}_{j}$, for all~$1 \leq j \leq n$ as follows.
Intuitively, the formula~$\chi^{a}_{j}$ is true in exactly
those worlds
from which there is an alternating chain that ends in a $z_0$-world,
that is of length at least~$j$
and that starts with an $a$-relation.
Similarly, the formula~$\chi^{b}_{j}$ is true in exactly
those worlds
from which there is an alternating chain that ends in a $z_0$-world,
that is of length at least~$j$,
and that starts with a $b$-relation.
We let~$\chi^{a}_{0} = \chi^{b}_{0} = z_0$,
and for each~$j > 0$, we let~$\chi^{b}_{j} = \hat{K}_{b}
(\neg z_1 \wedge \neg z_2 \wedge \chi^{a}_{j-1})$
and~$\chi^{a}_{j} = \hat{K}_{a}
(\neg z_1 \wedge \neg z_2 \wedge \chi^{b}_{j-1})$.
Now, using the formulas~$\chi^{b}_j$,
we can define the formulas~$\chi_j$.
We let~$\chi_j = z_1 \wedge \neg z_2 \wedge \chi^{b}_j \wedge \neg\chi^{b}_{j-1}$.
As a result of this definition,
the formula~$\chi_j$ is true in exactly those worlds
that are the first world of a chain of length exactly~$j$.

For example, consider the formula~$\chi_2 = z_1 \wedge \neg z_2 \wedge
\hat{K}_b (\neg z_1 \wedge \neg z_2 \wedge
\hat{K}_a (\neg z_1 \wedge \neg z_2 \wedge z_0)) \wedge
\neg \hat{K}_b (\neg z_1 \wedge \neg z_2 \wedge z_0)$
and consider the group of worlds depicted in
Figure~\ref{fig:group-example}. This formula is true only in the
first world of the chain of length~$2$.

The epistemic model~$(\MMM,w_0)$ that we use in the
reduction is based on the model~$\MMM_{\alpha_0}$
representing the truth
assignment~$\alpha_0 : \SBs x_1,\dotsc,x_n \SEs \rightarrow
\SBs 0,1 \SEs$ that sets all propositions~$x_1,\dotsc,x_n$ to true.
To obtain~$\MMM$, we will add a number of additional worlds
to the model~$\MMM_{\alpha_0}$,
that we will use to simulate the behavior of the
existential and universal quantifiers in the DEL-formula~$\chi$
that we will construct below.
Specifically, we will add alternating chains of worlds to the
model that are similar to the chains that represent the
propositions~$x_1,\dotsc,x_n$.
However, the additional chains that we add differ in two aspects
from the chains that represent the propositions~$x_1,\dotsc,x_n$:
(1)~the additional chains start
with an $a$-relation instead of starting with a $b$-relation,
and (2)~in the first world of the additional chains, the
propositional variable~$z_2$ is true instead of the variable~$z_1$.
For each~$1 \leq i \leq n$, we add such an additional chain of
length~$i$, and we connect the first worlds of these additional chains,
together with the designated world, in a clique of $b$-relations.
To illustrate this,
the model~$(\MMM,w_0)$ that results from this construction
is shown in Figure~\ref{fig:initial-model-for-three-vars},
for the case where~$n = 3$.
For the sake of convenience, we will denote these additional
chains by \emph{$z_2$-chains}, and the chains that represent
the propositions~$x_1,\dotsc,x_n$ by \emph{$z_1$-chains}
(after the propositional variables that are true in the first worlds
of these chains).

%The formulas~$\chi_j$ that we defined above can be used
%to establish whether an alternating chain of length exactly~$j$,
%that starts from a $z_1$-world with a $b$-relation,
%is present in the model.
To check whether an alternating chain of length
exactly~$j$, that starts from a $z_2$-world
with an $a$-relation, is present in the model,
we define formulas~$\chi'_j$
similarly to the way we defined
the formulas~$\chi_j$. Specifically, we
let~$\chi'_j = \neg z_1 \wedge z_2 \wedge \chi^{a}_j \wedge
\neg\chi^{a}_{j-1}$.

We will use the $z_2$-chains 
%starting from a $z_2$-world with an $a$-relation---%
together with the formulas~$\chi'_j$
%---%
to keep track of an additional counter.
We will use this counter 
%we will use 
as a technical trick to implement
the simulation of existentially and universally quantified variables
in the 
%quantified Boolean 
formula~$\varphi$. 
%using
%the combination of event models and the DEL formula
%that we will construct.

\begin{figure}[htp!]
\begin{center}
\begin{tikzpicture}
  \tikzstyle{dnode}=[inner sep=1pt,outer sep=1pt,draw,circle,minimum width=9pt]
  \tikzstyle{nnode}=[inner sep=1pt,outer sep=1pt,circle,minimum width=9pt]
  \tikzstyle{label-edge}=[midway,fill=white, inner sep=1pt]
  % CHAIN 1
  \node[nnode, label=above:{\scriptsize $z_1$}] (u0) at (0.25,0) {};
  \fill (u0) circle [radius=2pt];
  \node[nnode, label=above:{\scriptsize $z_0$}] (u1) at (1.25,0) {};
  \fill (u1) circle [radius=2pt];
  \draw[-] (u0) -- (u1) node[label-edge] {\scriptsize $b$};
  % CHAIN 2
  \node[nnode, label={[xshift=5pt]above:{\scriptsize $z_1$}}] (v0) at (1.5,-1.5) {};
  \fill (v0) circle [radius=2pt];
  \node[nnode] (v1) at (2.5,-1.5) {};
  \fill (v1) circle [radius=2pt];
  \node[nnode, label=above:{\scriptsize $z_0$}] (v2) at (3.5,-1.5) {};
  \fill (v2) circle [radius=2pt];
  \draw[-] (v0) -- (v1) node[label-edge] {\scriptsize $b$};
  \draw[-] (v1) -- (v2) node[label-edge] {\scriptsize $a$};
  % CHAIN 4
  \node[nnode, label={below:{\scriptsize $z_1$}}] (w0) at (0.25,-3) {};
  \fill (w0) circle [radius=2pt];
  \node[nnode] (w1) at (1.25,-3) {};
  \fill (w1) circle [radius=2pt];
  \node[nnode] (w2) at (2.25,-3) {};
  \fill (w2) circle [radius=2pt];
  \node[nnode, label=below:{\scriptsize $z_0$}] (w3) at (3.25,-3) {};
  \fill (w3) circle [radius=2pt];
  \draw[-] (w0) -- (w1) node[label-edge] {\scriptsize $b$};
  \draw[-] (w1) -- (w2) node[label-edge] {\scriptsize $a$};
  \draw[-] (w2) -- (w3) node[label-edge] {\scriptsize $b$};
  % CLIQUE
  \draw[-] (u0) -- (v0) node[label-edge] {\scriptsize $a$};
  \draw[-] (u0) -- (w0) node[label-edge, near start] {\scriptsize $a$};
  \draw[-] (v0) -- (w0) node[label-edge] {\scriptsize $a$};
  \node[dnode, label={[yshift=7pt]above:{\scriptsize $z_1,z_2$}}] (z) at (-1,-1.5) {};
  \fill (z) circle [radius=2pt];
  \draw[-] (z) -- (u0) node[label-edge, near end] {\scriptsize $a$};
  \draw[-] (z) -- (v0) node[label-edge, near end] {\scriptsize $a$};
  \draw[-] (z) -- (w0) node[label-edge] {\scriptsize $a$};
  \node[nnode, label=above:{\scriptsize $z_2$}] (z0) at (-2.25,0) {};
  \fill (z0) circle [radius=2pt];
  \node[nnode, label={[xshift=-5pt]above:{\scriptsize $z_2$}}] (z1) at (-3.5,-1.5) {};
  \fill (z1) circle [radius=2pt];
  \node[nnode, label=below:{\scriptsize $z_2$}] (z2) at (-2.25,-3) {};
  \fill (z2) circle [radius=2pt];
  \draw[-] (z) -- (z0) node[label-edge, near end] {\scriptsize $b$};
  \draw[-] (z) -- (z1) node[label-edge, near end] {\scriptsize $b$};
  \draw[-] (z) -- (z2) node[label-edge] {\scriptsize $b$};
  \draw[-] (z0) -- (z1) node[label-edge] {\scriptsize $b$};
  \draw[-] (z1) -- (z2) node[label-edge] {\scriptsize $b$};
  \draw[-] (z0) -- (z2) node[label-edge, near start] {\scriptsize $b$};
  % CHAIN 1'
  \node[nnode, label=above:{\scriptsize $z_0$}] (z0') at (-3.25,0) {};
  \fill (z0') circle [radius=2pt];
  \draw[-] (z0) -- (z0') node[label-edge] {\scriptsize $a$};
  % CHAIN 2'
  \node[nnode] (z1') at (-4.5,-1.5) {};
  \fill (z1') circle [radius=2pt];
  \node[nnode, label=above:{\scriptsize $z_0$}] (z1'') at (-5.5,-1.5) {};
  \fill (z1'') circle [radius=2pt];
  \draw[-] (z1) -- (z1') node[label-edge] {\scriptsize $a$};
  \draw[-] (z1') -- (z1'') node[label-edge] {\scriptsize $b$};
  % CHAIN 3'
  \node[nnode, label=below:{\scriptsize $z_1$}] (z2') at (-3.25,-3) {};
  \fill (z2') circle [radius=2pt];
  \node[nnode] (z2'') at (-4.25,-3) {};
  \fill (z2'') circle [radius=2pt];
  \node[nnode, label=below:{\scriptsize $z_0$}] (z2''') at (-5.25,-3) {};
  \fill (z2''') circle [radius=2pt];
  \draw[-] (z2) -- (z2') node[label-edge] {\scriptsize $a$};
  \draw[-] (z2') -- (z2'') node[label-edge] {\scriptsize $b$};
  \draw[-] (z2'') -- (z2''') node[label-edge] {\scriptsize $a$};
\end{tikzpicture}
\end{center}
\caption{The epistemic model~$(\MMM,w_0)$
for the case where~$n = 3$, as used
in the proof of Theorem~\ref{thm:pspace-hardness}.}
\label{fig:initial-model-for-three-vars}
\end{figure}
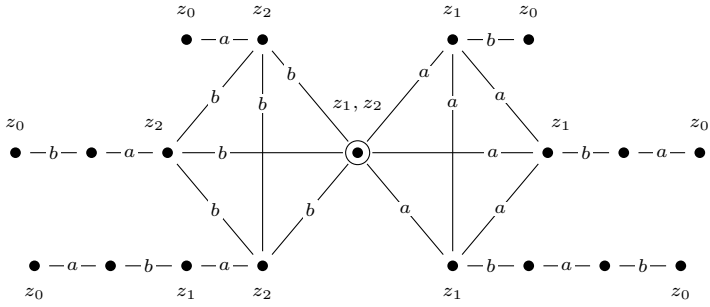

Next, we describe how we can generate all possible
truth assignments over
the variables~$x_1,\dotsc,x_n$ from the initial model~$\MMM$.
We do this in such a way that we
can afterwards express the existential and universal
quantifications of the formula~$\varphi$
using modal operators in the epistemic language.
In order to generate groups of worlds that represent truth
assignments~$\alpha$ that differ from the all-ones assignment~$\alpha_0$,
we will apply updates that copy the existing worlds
%(and that therefore copy
%all existing groups representing some truth assignments)
but that eliminate
(the first worlds of) chains of a certain length.
This is the third main idea behind this reduction.

Specifically, we will introduce a single-pointed event
model~$(\EEE_i,e_i)$ for each propositional variable~$x_i$,
that is depicted in Figure~\ref{fig:update-xi}.
Intuitively, what happens when the update~$(\EEE_i,e_i)$ is applied
is the following.
All existing groups of worlds will be duplicated, resulting in five copies%
---this corresponds to the five
events~$f^1_i,\dotsc,f^5_i$ in the event model.
The resulting groups of worlds will be connected corresponding to the
relations between the events in the event model. That is, for any existing
group of worlds, three of its copies (corresponding to the
events~$f^1_i$,~$f^2_i$, and~$f^3_i$) will be connected by $b$-relations.
The second and third of these copies (the ones corresponding
to the events~$f^2_i$ and~$f^3_i$) will be connected by $a$-relations to the
fourth and fifth copy (corresponding to the events~$f^4_i$ and~$f^5_i$),
respectively.
Moreover, in the fourth and fifth copy,
(the first world of) the $z_2$-chain of length~$i$ is removed,
and in the fifth copy, (the first world of) the $z_1$-chain of
length~$i$ is removed as well.
These effects of removing (the first worlds of) chains
is enforced by the preconditions
of the events in the event model.

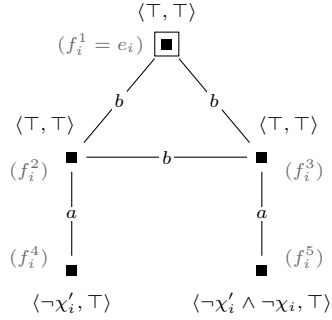
\begin{figure}[htp!]
\begin{center}
\begin{tikzpicture}
  \tikzstyle{dnode}=[inner sep=1pt,outer sep=1pt,draw,circle,minimum width=9pt]
  \tikzstyle{nnode}=[inner sep=1pt,outer sep=1pt,circle,minimum width=9pt]
  \tikzstyle{dnode'}=[inner sep=1pt,outer sep=1pt,draw,rectangle,minimum width=9pt,minimum height=9pt]
  \tikzstyle{nnode'}=[inner sep=1pt,outer sep=1pt,rectangle,minimum width=9pt,minimum height=9pt]
  \tikzstyle{label-edge}=[midway,fill=white, inner sep=1pt]
  % NODES
  \node[dnode', label={[xshift=0pt, yshift=0pt]left:{\textcolor{black!60}{\scriptsize $(f^1_i = e_i)$}}}, label=above:{\scriptsize $\langle \top,\top \rangle$}] (v0) at (3.25,1.5) {};
  \fill {(v0)+(-0.07,-0.07)} rectangle (3.32,1.57);
  \node[nnode', label={[yshift=-5pt]left:{\textcolor{black!60}{\scriptsize $(f^2_i)$}}}, label={[xshift=-10pt]above:{\scriptsize $\langle \top, \top \rangle$}}] (v1) at (2,0) {};
  \fill {(v1)+(-0.07,-0.07)} rectangle (2.07,0.07);
  \node[nnode', label={[yshift=-5pt]right:{\textcolor{black!60}{\scriptsize $(f^3_i)$}}}, label={[xshift=10pt]above:{\scriptsize $\langle \top, \top \rangle$}}] (v2) at (4.5,0) {};
  \fill {(v2)+(-0.07,-0.07)} rectangle (4.57,0.07);
  \node[nnode', label={[yshift=5pt]left:{\textcolor{black!60}{\scriptsize $(f^4_i)$}}}, label=below:{\scriptsize $\langle \neg\chi'_{i}, \top \rangle$}] (w1) at (2,-1.5) {};
  \fill {(w1)+(-0.07,-0.07)} rectangle (2.07,-1.43);
  \node[nnode', label={[yshift=5pt]right:{\textcolor{black!60}{\scriptsize $(f^5_i)$}}}, label=below:{\scriptsize $\langle \neg\chi'_{i} \wedge \neg\chi_{i}, \top \rangle$}] (w2) at (4.5,-1.5) {};
  \fill {(w2)+(-0.07,-0.07)} rectangle (4.57,-1.43);
  % EDGES
  \draw[-] (v0) -- (v1) node[label-edge] {\scriptsize $b$};
  \draw[-] (v0) -- (v2) node[label-edge] {\scriptsize $b$};
  \draw[-] (v1) -- (v2) node[label-edge] {\scriptsize $b$};
  \draw[-] (v1) -- (w1) node[label-edge] {\scriptsize $a$};
  \draw[-] (v2) -- (w2) node[label-edge] {\scriptsize $a$};
\end{tikzpicture}
\end{center}
\caption{The event model~$(\EEE_i,e_i)$ corresponding to
the propositional variable~$x_i$,
used in the proof of Theorem~\ref{thm:pspace-hardness}.
The events are labelled~$f^1_i,\dotsc,f^5_i$.}
\label{fig:update-xi}
\end{figure}

By applying the updates~$(\EEE_1,e_1),\dotsc,(\EEE_n,e_n)$, we generate
many (in fact, an exponential number of)
copies of the model~$\MMM$, in each of which certain chains of
worlds are removed, and which are connected to each other by means of
$a$-relations and $b$-relations in the way described in the previous paragraph.
In particular, for each truth assignment~$\alpha$ to the
propositions~$x_1,\dotsc,x_n$,
there is some group of worlds that corresponds to~$\alpha$.

Finally, we construct the DEL-formula~$\xi$.
We let~$\xi = [\EEE_1,e_1]\dotsc[\EEE_n,e_n] \xi_{1}$,
where we define~$\xi_{1}$ below.
The formula~$\xi_{1}$ exploits the structure of the epistemic
model~$\MMM'$,
that results from updating the model~$\MMM$
with the updates~$(\EEE_i,e_i)$,
to simulate the semantics of the quantified Boolean formula~$\varphi$.
For each~$1 \leq i \leq n+1$, we define~$\xi_{i}$ inductively
as follows:
\[ \begin{array}{r l}
  \xi_{i} = & \begin{dcases*}
    \psi' & if~$i = n+1$, \\
    \hat{K}_b \hat{K}_a (z_1 \wedge z_2 \wedge \bigwedge\limits_{\mathllap{1 \leq} j \mathrlap{\leq i}} \neg\hat{K}_b \chi'_{j} \wedge \bigwedge\limits_{\mathllap{i <} j \mathrlap{\leq n}} \hat{K}_b \chi'_{j} \wedge \xi_{i+1}) &
      for odd~$i \leq n$, \\
    K_b K_a ((z_1 \wedge z_2 \wedge \bigwedge\limits_{\mathllap{1 \leq} j \mathrlap{\leq i}} \neg\hat{K}_b \chi'_{j} \wedge \bigwedge\limits_{\mathllap{i <} j \mathrlap{\leq n}} \hat{K}_b \chi'_{j}) \rightarrow \xi_{i+1}) &
      for even~$i \leq n$. \\
  \end{dcases*}
\end{array} \]
Here,~$\psi'$ is the formula obtained from~$\psi$
(the quantifier-free part of the
quantified Boolean formula~$\varphi$) by replacing each occurrence
of a propositional variable~$x_i$ by the formula~$\hat{K}_{a} \chi_i$.

We use the formulas~$\xi_{i}$ to express
the formula~$\varphi$ with its existentially and universally quantified
variables.
Intuitively, the formulas~$\xi_{i}$ navigate through the groups of
worlds in the
model~$\MMM'$---resulting from updating~$\MMM$ with the event
models~$(\EEE_i,e_i)$---%
as follows.
Consider the central world of some group of worlds
in the model~$\MMM'$,
and consider the formula~$\xi_i$ for some odd~$i \leq n$.
For~$\xi_i$ to be true in this world,
the formula~$\xi_{i+1}$ needs to be true in the central world
of \emph{some} group of worlds
%that is a copy of this group of worlds
that corresponds to one of the events~$f^4_i$ or~$f^5_i$
from the event model~$(\EEE_i,e_i)$.
Similarly, for the formula~$\xi_i$ to be true in this world,
for even~$i \leq n$,
the formula~$\xi_{i+1}$ needs to be true in the central world
of \emph{both} groups of worlds
%that are a copy of this group of worlds
that correspond to the events~$f^4_i$ and~$f^5_i$.
In this way, for odd~$i \leq n$, the formula~$\xi_i$ together with the
event model~$(\EEE_i,e_i)$ serves to simulate an existential choice
of a truth value for the variable~$x_i$.
Similarly, for even~$i \leq n$, the formula~$\xi_i$ together with the
event model~$(\EEE_i,e_i)$ serves to simulate a universal choice
of a truth value for the variable~$x_i$.

The model~$(\MMM,w_0)$ and the formula~$\xi$ can be constructed in
polynomial time in the size of the quantified Boolean formula~$\varphi$.
Furthermore, in the constructed instance,
there are only two agents,
the epistemic model~$(\MMM,w_0)$ is a single-pointed S5 model,
all event models~$(\EEE_i,e_i)$ are single-pointed S5
models without postconditions,
and~$\xi_1$ is a formula without update modalities that contains
(many occurrences of) only three propositional variables~$z_0,z_1,z_2$.

We show that~$\varphi$ is a true
quantified Boolean formula if and only if~$\MMM,w_0 \models \xi$.
In order to do so, we prove the following (technical) statement relating truth
assignments~$\alpha$ to the propositions~$x_1,\dotsc,x_n$
and (particular) worlds~$w$ in the epistemic model~$(\MMM',w'_0) =
(\MMM,w_0) \otimes (\EEE_1,e_1) \otimes \dotsm \otimes (\EEE_n,e_n)$.
Before we give the statement that we will prove, we observe that
every world~$w$ that sets both~$z_1$ and~$z_2$
to true is the central world of some
group of worlds that represents a truth assignment~$\alpha$
to the propositions~$x_1,\dotsc,x_n$.
For the sake of convenience, we will say
that~$w$ corresponds to the truth assignment~$\alpha$.
The statement that we will prove for all~$1 \leq i \leq n+1$ is the following.

\medskip
\noindent \textit{Statement:}
Let~$\alpha$ be any truth assignment to the propositions~$x_1,\dotsc,x_{i-1}$.
Moreover, let~$w$ be any world in the model~$(\MMM',w'_0)$ such that:
\begin{enumerate}
  \item $w$ makes~$z_1$ and~$z_2$ true,
  \item $w$ makes~$\hat{K}_{b} \chi'_{j}$ false for all~$1 \leq j < i$,
  \item $w$ makes~$\hat{K}_{b} \chi'_{j}$ true
    for all~$i \leq j \leq n$, and
  \item the truth assignment corresponding to~$w$ agrees
    with~$\alpha$ on the propositions~$x_1,\dotsc,x_{i-1}$.
\end{enumerate}
Then the (partially) quantified Boolean formula~$Q_i x_i \dotsc \exists x_{n-1}
\forall x_n. \psi$ is true under~$\alpha$ if and only
if~$w$ makes~$\xi_i$ true.
\medskip

Observe that for~$i = 1$, the world~$w'_0$ satisfies all four conditions.
Therefore, the statement for~$i = 1$ implies
that~$\MMM,w_0 \models \xi$
if and only if~$\varphi$ is a true quantified Boolean formula.
Thus, proving this statement for all~$1 \leq i \leq n+1$ suffices
to show the correctness of our reduction.

We begin by showing that the statement holds for~$i = n + 1$.
In this case, we know that~$\alpha$ is a truth assignment to the
propositions~$x_1,\dotsc,x_n$.
Moreover,~$\xi_{n+1} = \psi'$.
%Because the model~$\MMM'$ is formed using product update
%with the event model~$\EEE_n$,
%we know that~$w$ makes~$\hat{K}_a \chi_n$ true.
%However,~$x_n$ does not occur in~$\psi'$.
%Therefore,
By construction of~$\psi'$, we know
that~$w$ makes~$\psi'$ true if and only if~$\alpha$ satisfies~$\psi$.
Therefore, the statement holds for~$i = n + 1$.

Next, we let~$1 \leq i \leq n$ be arbitrary,
and we assume that the statement holds for~$i+1$.
That is, the statement holds for every combination of a truth
assignment~$\alpha$ and a world~$w$ that satisfies the conditions.
Since~$w$ is a world in the model~$(\MMM,w_0) \otimes
(\EEE_1,e_1) \otimes \dotsm \otimes (\EEE_n,e_n)$,
and since~$w$ makes~$z_1$ and~$z_2$ true,
we know that~$w = (w_0,e'_1,\dotsc,e'_{n})$,
for some~$e'_1,\dotsc,e'_n$,
where~$e'_j \in \EEE_j$ for all~$1 \leq j \leq n$.
We know that~$w$ makes~$\hat{K}_b \chi'_{j}$ true for all~$i \leq j \leq n$.
Therefore, we know that for each~$i \leq j \leq n$
it holds that~$e'_j \in \SBs f^1_j, f^2_j, f^3_j \SEs$.

We now distinguish two cases: either (1)~$i$ is odd,
or (2)~$i$ is even.
In case~(1), the $i$-th
quantifier of~$\varphi$ is existential,
and in case~(2), the $i$-th quantifier of~$\varphi$
is universal.
First, consider case~(1).
Suppose that~$\exists x_i \dotsc \exists x_{n-1} \forall x_n. \psi$ is true
under~$\alpha$. Then there exists some truth assignment~$\alpha'$
to the propositions~$x_1,\dotsc,x_i$ that agrees with~$\alpha$
on the propositions~$x_1,\dotsc,x_{i-1}$
and that ensures that~$\forall x_{i+1} \dotsc
\exists x_{n-1} \forall x_n. \psi$ is true under~$\alpha'$.
Suppose that~$\alpha'(x_i) = 0$; the case for~$\alpha'(x_i) = 1$
is entirely similar.
Now, consider the
worlds~$w' = (w_0,e'_1,\dotsc,e'_{i-1},f^3_i,e'_{i+1},\dotsc,e'_{n})$
and~$w'' = (w_0,e'_1,\dotsc,e'_{i-1},f^5_i,e'_{i+1},\dotsc,e'_{n})$.
By the construction of~$(\EEE_i,e_i)$, by the semantics
of product update, and by the fact
that~$e'_i \in \SBs f^1_i,f^2_i,f^3_i \SEs$,
it holds that~$(w,w') \in R_b$ and~$(w',w'') \in R_a$.
Moreover, it is straightforward to verify that~$w''$ satisfies
conditions~(1)--(4), for the truth assignment~$\alpha'$.
Also, we know that~$\forall x_{i+1} \dotsc
\exists x_{n-1} \forall x_n. \psi$ is true under~$\alpha'$.
Therefore, by the induction hypothesis,
we know that~$w''$ makes the formula~$\xi_{i+1}$ true.
It then follows from the definition of~$\xi_i$
that~$w'$ and~$w''$ witness that~$w$ makes~$\xi_i$ true.

Conversely, suppose that~$w$ makes~$\xi_i$ true.
Moreover, suppose that~$w'$ and~$w''$ (as defined above) witness this.
(The only other possible worlds~$u'$ and~$u''$ that could
witness this are obtained from~$w$ by replacing~$e'_i$ by~$f^2_i$
and~$f^4_i$, respectively. The case where~$u'$ and~$u''$ witness
that~$w$ makes~$\xi_i$ true is entirely similar.)
This means that~$w''$ makes~$\xi_{i+1}$ true.
Then, by the induction hypothesis, it follows that the truth
assignment~$\alpha'$ to the propositions~$x_1,\dotsc,x_i$ corresponding
to the world~$w''$
satisfies~$\forall x_{i+1} \dotsc \exists x_{n-1} \forall x_n. \psi$.
Moreover, since~$\alpha'$ agrees with~$\alpha$ on the
propositions~$x_1,\dotsc,x_{i-1}$, it follows
that~$\exists x_{i} \dotsc \exists x_{n-1} \forall x_n. \psi$
is true under~$\alpha$.

\sloppypar
Next, consider case~(2).
Suppose that~$\forall x_i \dotsc \exists x_{n-1} \forall x_n. \psi$
is true under~$\alpha$.
Then for both truth assignments~$\alpha'$
to the variables~$x_1,\dotsc,x_i$ that agree with~$\alpha$
it holds that~$\exists x_{i+1} \dotsc
\exists x_{n-1} \forall x_n. \psi$ is true under~$\alpha'$.
The only worlds that satisfy~$(z_1 \wedge z_2 \wedge
\bigwedge\nolimits_{1 \leq j \leq i} \neg\hat{K}_b \chi'_{j} \wedge
\bigwedge\nolimits_{i < j \leq n} \hat{K}_b \chi'_{j})$
and that are accessible from~$w$
by a $b$-relation followed by an $a$-relation
are the worlds~$u_1$ and~$u_2$,
where~$u_1 = (w_0,e'_1,\dotsc,e'_{i-1},f^4_i,e'_{i+1},\dotsc,e'_{n})$
and~$u_2 = (w_0,e'_1,\dotsc,e'_{i-1},f^5_i,e'_{i+1},\dotsc,e'_{n})$.
Moreover, the truth assignments~$\alpha_1$
and~$\alpha_2$ that correspond to~$u_1$ and~$u_2$,
respectively,
agree with~$\alpha$ on the propositions~$x_1,\dotsc,x_{i-1}$.
Because~$\forall x_i \dotsc \exists x_{n-1} \forall x_n. \psi$
is true under~$\alpha$, we know
that~$\exists x_{i+1} \dotsc \exists x_{n-1} \forall x_n. \psi$
is true under both~$\alpha_1$ and~$\alpha_2$.
Then, by the induction hypothesis
it follows that both~$u_1$ and~$u_2$ make~$\xi_{i+1}$ true.
Therefore,~$w$ makes~$\xi_i$ true.

Conversely, suppose that~$w$ makes~$\xi_i$ true.
By the definition of~$\xi_i$, we then know that all worlds that
are accessible from~$w$ by a $b$-relation followed by
an $a$-relation and that make~$(z_1 \wedge z_2 \wedge
\bigwedge\nolimits_{1 \leq j \leq i} \neg\hat{K}_b \chi'_{j} \wedge
\bigwedge\nolimits_{i < j \leq n} \hat{K}_b \chi'_{j})$ true,
also make~$\xi_{i+1}$ true.
Consider the worlds~$u_1$ and~$u_2$, as defined above.
These are both accessible from~$w$ by a $b$-relation followed
by an $a$-relation, and they make~$(z_1 \wedge z_2 \wedge
\bigwedge\nolimits_{1 \leq j \leq i} \neg\hat{K}_b \chi'_{j} \wedge
\bigwedge\nolimits_{i < j \leq n} \hat{K}_b \chi'_{j})$ true.
Therefore, both~$u_1$ and~$u_2$ make~$\xi_{i+1}$ true.
Also, the truth assignments~$\alpha_1$ and~$\alpha_2$
that correspond to~$u_1$ and~$u_2$, respectively,
agree with~$\alpha$ on the propositions~$x_1,\dotsc,x_{i-1}$.
Moreover, the truth assignments~$\alpha_1$ and~$\alpha_2$
are both possible truth assignments to the
propositions~$x_1,\dotsc,x_i$ that agree with~$\alpha$.
By the induction hypothesis, the
formula~$\exists x_{i+1} \dotsc \exists x_{n-1} \forall x_n. \psi$ is true
under both~$\alpha_1$ and~$\alpha_2$.
Therefore, we can conclude
that~$\forall x_i \dotsc \exists x_{n-1} \forall x_n. \psi$
is true under~$\alpha$.

This concludes the inductive proof of the statement for
all~$1 \leq i \leq n+1$,
and thus concludes our correctness proof.
Therefore, we can conclude that the problem is \PSPACE{}-hard.
\qed\end{proof}

%%%
%%% RESULTS FOR SEMI-PRIVATE ANNOUNCEMENTS
%%%
\section{Results for semi-private announcements}
\label{sec:semi-private-announcements}

Next, we consider the model checking problem for DEL when
restricted to updates that are semi-private announcements.
In fact, \PSPACE{}-hardness for the setting with semi-private
announcements (rather than allowing arbitrary event models)
already follows from a recent \PSPACE{}-hardness proof for a restricted
variant of the model checking problem---{ see~\cite[Theorem~4]{VandePol15}
and~\cite[Theorem~1]{VandePolVanRooijSzymanik18}.}
In that \PSPACE{}-hardness result, the number of agents is unbounded,
i.e., the number of agents is part of the problem input.
We show that the problem is already \PSPACE{}-hard
when the number of agents is bounded
by any constant~$k \geq 2$.

{
The result of Theorem~\ref{thm:semi-private-pspace-hardness}
shows that the inherent hardness of the model checking problem
for DEL---which we saw in Theorem~\ref{thm:pspace-hardness}
is already present in a very restricted setting---is even present when
we restrict event models to be of a very specific shape (i.e.,
semi-private announcements).
}

Theorem~\ref{thm:semi-private-pspace-hardness} is a stronger
result than Theorem~\ref{thm:pspace-hardness}---%
Theorem~\ref{thm:semi-private-pspace-hardness}
implies the result of Theorem~\ref{thm:pspace-hardness}.
We presented the proof of Theorem~\ref{thm:pspace-hardness} in full
detail, because it allows us to explain the proof of
Theorem~\ref{thm:semi-private-pspace-hardness} in a clear way.
{
In fact, the result of Theorem~\ref{thm:semi-private-pspace-hardness}
is stronger than all other \PSPACE{}-hardness results
for the model checking problem for DEL in the literature---%
i.e.,~%
\cite[Proposition~7.2]{VanEijckSchwarzentruber14},
~\cite[Theorem~2]{AucherSchwarzentruber13} and~%
\cite{VandePol15,VandePolVanRooijSzymanik18}.
These results all depend on allowing certain parts of the DEL setting
being unrestricted---e.g., allowing multi-pointed models,
allowing more than two agents, or allowing non-serial relations---%
and their proofs cannot easily be modified to work for the more
restricted setting of Theorem~\ref{thm:semi-private-pspace-hardness}
with (single-pointed) semi-private announcements, S5 relations,
only two agents, and only three propositional variables.
}

{
The main technical hurdle that needs to be overcome to establish
Theorem~\ref{thm:semi-private-pspace-hardness} is to simulate
the event models that we used in the proof
of Theorem~\ref{thm:pspace-hardness} (each consisting of~5 events)
by using a sequence of semi-private announcements
(which are event models with~2 events).
{
We do this by constructing---for each event model---three semi-private
announcements, in such a way
that when these three semi-private announcements are
composed, they take the role of the event model in the proof.
In order to make sure that the semi-private announcements
correctly take over the role of the event models in the proof,
we also need to adapt the part of the DEL formula that expresses
(the unquantified part of) the quantified Boolean formula.
}
}

\begin{theorem}
\label{thm:semi-private-pspace-hardness}
The model checking problem for DEL is \PSPACE{}-hard,
even when restricted to the case where the question is
whether~$\MMM,w_0 \models [\EEE_1,e_1]\dotsc[\EEE_n,e_n] \chi$,
where:
\begin{itemize}
  \item the model~$(\MMM,w_0)$ is a single-pointed S5 model;
  \item all the~$(\EEE_i,e_i)$ are (single-pointed) semi-private
    announcements;
  \item $\chi$ is an epistemic formula without update modalities
    that contains (multiple occurrences of)
    only three propositional variables; and
  \item there are only two agents.
\end{itemize}
\end{theorem}
\begin{proof}
We modify the proof of Theorem~\ref{thm:pspace-hardness} to
work also for the case of semi-private announcements.
Most prominently, we will replace the event models~$(\EEE_i,e_i)$
that are used in the proof of Theorem~\ref{thm:pspace-hardness}
(shown in Figure~\ref{fig:update-xi}) by a number of event models
for semi-private announcements.
Intuitively, these semi-private announcements will take the role
of the event models~$(\EEE_i,e_i)$.
In order to make this work, we will also slightly change the
initial model~$\MMM$.

As in the proof of Theorem~\ref{thm:pspace-hardness},
we give a polynomial-time reduction
from the problem of deciding whether a quantified Boolean formula is true.
Let~$\varphi = \exists x_1 \forall x_2 \dotsc \exists x_{n-1} \forall x_{n}. \psi$
be a quantified Boolean formula, where~$\psi$ is quantifier-free.
%(We may assume without loss of generality that~$\psi$ does not contain
%any occurrences of~$x_n$. If this were not the case, we could simply add
%dummy variables~$x_{n+1},x_{n+2}$ to ensure this.
%We will use this assumption later in the proof for technical reasons.)
We construct an epistemic model~$(\MMM,w_0)$ with two agents~$a,b$
and a DEL-formula~$\xi$ such that~$\MMM,w_0 \models \xi$ if and only
if~$\varphi$ is true.

In the proof of Theorem~\ref{thm:pspace-hardness},
the initial model~$\MMM$ consisted of a central world
(where~$z_1$ and~$z_2$ are true),
a number of $z_1$-chains (for each~$1 \leq i \leq n$,
there is a $z_1$-chain of length~$i$),
and a number of $z_2$-chains (for each~$1 \leq i \leq n$,
there is a $z_2$-chain of length~$i$)---%
and these worlds are connected by $a$-relations and
$b$-relations as shown in
Figure~\ref{fig:initial-model-for-three-vars}.
To obtain the initial model~$\MMM$ that we use in this proof,
we add additional $z_2$-chains.
Specifically, for each~$1 \leq i \leq 3n$,
we will have a $z_2$-chain of length~$i$.
These additional $z_2$-chains are connected to the central
world in exactly the same way as the original $z_2$-chains
(that is, all the first worlds of the $z_2$-chains are connected
in a $b$-clique to the central world).
We will use these additional $z_2$-chains to simulate the behavior
of the event models~$(\EEE_i,e_i)$ from the proof of
Theorem~\ref{thm:pspace-hardness} with
event models corresponding to semi-private announcements.
The number of $z_1$-chains remains the same.
The designated world~$w_0$ is the central world (that is, the
only world that makes both~$z_1$ and~$z_2$ true),
as in the proof of Theorem~\ref{thm:pspace-hardness}.

\begin{figure}[htp!]

%%%
\begin{subfigure}[b]{\textwidth}
\centering
\begin{tikzpicture}
  \tikzstyle{dnode}=[inner sep=1pt,outer sep=1pt,draw,circle,minimum width=9pt]
  \tikzstyle{nnode}=[inner sep=1pt,outer sep=1pt,circle,minimum width=9pt]
  \tikzstyle{dnode'}=[inner sep=1pt,outer sep=1pt,draw,rectangle,minimum width=9pt,minimum height=9pt]
  \tikzstyle{nnode'}=[inner sep=1pt,outer sep=1pt,rectangle,minimum width=9pt,minimum height=9pt]
  \tikzstyle{label-edge}=[midway,fill=white, inner sep=1pt]
  % NODES
  \node[] at (-3,0) {};
  \node[] at (6,0) {};
  \node[dnode', label=below:{\textcolor{black!60}{\scriptsize $(f^1_i)$}}, label=above:{\scriptsize $\langle \neg\chi'_{i+n}, \top \rangle$}] (w0) at (0,0) {};
  \fill {(w0)+(-0.07,-0.07)} rectangle (0.07,0.07);
  \node[nnode', label=below:{\textcolor{black!60}{\scriptsize $(f^2_i)$}}, label=above:{\scriptsize $\langle \neg\chi'_{i+2n}, \top \rangle$}] (w1) at (3,0) {};
  \fill {(w1)+(-0.07,-0.07)} rectangle (3.07,0.07);
  % EDGES
  \draw[-] (w0) -- (w1) node[label-edge] {\scriptsize $b$};
\end{tikzpicture}
\caption{The semi-private announcement~$(\EEE^1_i,f^1_i)$}
\label{fig:semi-private-hardness1}
\end{subfigure}

\bigskip

%%%
\begin{subfigure}[b]{\textwidth}
\centering
\begin{tikzpicture}
  \tikzstyle{dnode}=[inner sep=1pt,outer sep=1pt,draw,circle,minimum width=9pt]
  \tikzstyle{nnode}=[inner sep=1pt,outer sep=1pt,circle,minimum width=9pt]
  \tikzstyle{dnode'}=[inner sep=1pt,outer sep=1pt,draw,rectangle,minimum width=9pt,minimum height=9pt]
  \tikzstyle{nnode'}=[inner sep=1pt,outer sep=1pt,rectangle,minimum width=9pt,minimum height=9pt]
  \tikzstyle{label-edge}=[midway,fill=white, inner sep=1pt]
  % NODES
  \node[] at (-3,0) {};
  \node[] at (6,0) {};
  \node[dnode', label=below:{\textcolor{black!60}{\scriptsize $(f^3_i)$}}, label=above:{\scriptsize $\langle \top, \top \rangle$}] (w0) at (0,0) {};
  \fill {(w0)+(-0.07,-0.07)} rectangle (0.07,0.07);
  \node[nnode', label=below:{\textcolor{black!60}{\scriptsize $(f^4_i)$}}, label=above:{\scriptsize $\langle (\neg \hat{K}_b \chi'_{i+n} \wedge z_2) \rightarrow \neg\chi'_i, \top \rangle$}] (w1) at (3,0) {};
  \fill {(w1)+(-0.07,-0.07)} rectangle (3.07,0.07);
  % EDGES
  \draw[-] (w0) -- (w1) node[label-edge] {\scriptsize $a$};
\end{tikzpicture}
\caption{The semi-private announcement~$(\EEE^2_i,f^3_i)$}
\label{fig:semi-private-hardness2}
\end{subfigure}

\bigskip

%%%
\begin{subfigure}[b]{\textwidth}
\centering
\begin{tikzpicture}
  \tikzstyle{dnode}=[inner sep=1pt,outer sep=1pt,draw,circle,minimum width=9pt]
  \tikzstyle{nnode}=[inner sep=1pt,outer sep=1pt,circle,minimum width=9pt]
  \tikzstyle{dnode'}=[inner sep=1pt,outer sep=1pt,draw,rectangle,minimum width=9pt,minimum height=9pt]
  \tikzstyle{nnode'}=[inner sep=1pt,outer sep=1pt,rectangle,minimum width=9pt,minimum height=9pt]
  \tikzstyle{label-edge}=[midway,fill=white, inner sep=1pt]
  % NODES
  \node[] at (-3,0) {};
  \node[] at (6,0) {};
  \node[dnode', label=below:{\textcolor{black!60}{\scriptsize $(f^5_i)$}}, label=above:{\scriptsize $\langle \top, \top \rangle$}] (w0) at (0,0) {};
  \fill {(w0)+(-0.07,-0.07)} rectangle (0.07,0.07);
  \node[nnode', label=below:{\textcolor{black!60}{\scriptsize $(f^6_i)$}}, label=above:{\parbox{4.2cm}{\scriptsize $\langle ((\neg\hat{K}_b \chi'_{i+2n} \wedge z_2) \rightarrow \neg \chi'_i)\ \wedge$\\$\phantom{\langle}((\neg\hat{K}_a \hat{K}_b \chi'_{i+2n} \wedge z_1) \rightarrow \neg \chi_i), \top \rangle$}}] (w1) at (3,0) {};
  \fill {(w1)+(-0.07,-0.07)} rectangle (3.07,0.07);
  % EDGES
  \draw[-] (w0) -- (w1) node[label-edge] {\scriptsize $a$};
\end{tikzpicture}
\caption{The semi-private announcement~$(\EEE^3_i,f^5_i)$}
\label{fig:semi-private-hardness3}
\end{subfigure}

\caption{The semi-private announcements~$(\EEE^1_i,f^1_i)$,~$(\EEE^2_i,f^3_i)$ and~$(\EEE^3_i,f^5_i)$ used in the proof of Theorem~\ref{thm:semi-private-pspace-hardness}.}
\label{fig:semi-private-hardness}
\end{figure}
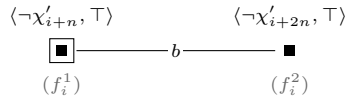
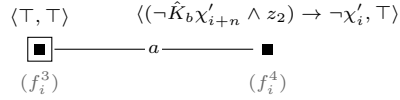
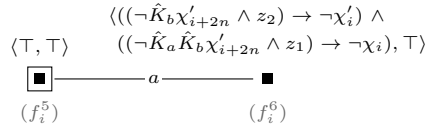

The event models~$(\EEE_i,e_i)$ that are used in the proof
of Theorem~\ref{thm:pspace-hardness} we replace
by the semi-private announcements~$(\EEE^1_i,f^1_i)$,%
~$(\EEE^2_i,f^3_i)$, and~$(\EEE^3_i,f^5_i)$,
as shown in Figure~\ref{fig:semi-private-hardness}.
The intuition behind these updates is the following.
Firstly, the semi-private announcement~$\EEE^1_i$,
shown in Figure~\ref{fig:semi-private-hardness1},
transforms every group of worlds into two copies,
and allows a choice between these two copies when
following a $b$-relation.
Moreover, in one copy, every (first world of the) $z_2$-chain
of length~$i+n$ is removed, and in the other copy,
every (first world of the) $z_2$-chain
of length~$i+2n$ is removed.
In other words, the choice between these two copies
determines whether the formula~$\hat{K}_b \chi'_{i+n}$
or the formula~$\hat{K}_b \chi'_{i+2n}$ is false
in the central world.
Then for every group of
worlds that does not include a $z_2$-chain of length~$i+n$, the semi-private announcement~$\EEE^2_i$,
shown in Figure~\ref{fig:semi-private-hardness2},
creates an $a$-accessible copy
where the $z_2$-chain of length~$i$ is removed.
Similarly, for every group of
worlds that does not include a $z_2$-chain of length~$i+2n$, the semi-private announcement~$\EEE^3_i$,
shown in Figure~\ref{fig:semi-private-hardness3},
creates an $a$-accessible copy
where both the $z_2$-chain of length~$i$
and the $z_1$-chain of length~$i$ are removed.

Next, we construct the DEL-formula~$\xi$.
We let~$\xi = [\EEE^1_1,f^1_1] [\EEE^2_1,f^3_1] [\EEE^3_1,f^5_1]
\dotsc \allowbreak{} [\EEE^1_n,f^1_n] [\EEE^2_n,f^3_n] [\EEE^3_n,f^5_n] \xi_{1}$,
where~$\xi_{1}$ is defined as follows,
similarly to the definition used in the proof of
Theorem~\ref{thm:pspace-hardness}.
For each~$1 \leq i \leq n+1$, we define~$\xi_{i}$ inductively
as follows:
\[ \begin{small} \begin{array}{r l}
  \xi_{i} = & \begin{dcases*}
    \psi' & if~$i = n+1$, \\
    \hat{K}_b (\bigwedge\limits_{\mathllap{1 \leq} j \mathrlap{\leq i}} \neg\hat{K}_b \chi'_{j} \wedge \hat{K}_a (z_1 \wedge z_2 \wedge \bigwedge\limits_{\mathllap{1 \leq} j \mathrlap{\leq i}} \neg\hat{K}_b \chi'_{j} \wedge \bigwedge\limits_{\mathllap{i <} j \mathrlap{\leq n}} \hat{K}_b \chi'_{j} \wedge \xi_{i+1})) &
      for odd~$i \leq n$, \\
    K_b ((\bigwedge\limits_{\mathllap{1 \leq} j \mathrlap{\leq i}} \neg\hat{K}_b \chi'_{j}) \rightarrow K_a ((z_1 \wedge z_2 \wedge \bigwedge\limits_{\mathllap{1 \leq} j \mathrlap{\leq i}} \neg\hat{K}_b \chi'_{j} \wedge \bigwedge\limits_{\mathllap{i <} j \mathrlap{\leq n}} \hat{K}_b \chi'_{j}) \rightarrow \xi_{i+1})) &
      for even~$i \leq n$. \\
  \end{dcases*}
\end{array} \end{small} \]
Here,~$\psi'$ is the formula obtained from~$\psi$
(the quantifier-free part of the
quantified Boolean formula~$\varphi$) by replacing each occurrence
of a propositional variable~$x_i$ by the formula~$\hat{K}_{a} \chi_i$.
%%
%% NOTE: give some more intuition about the formulas \xi_i,
%% in particular, where do they differ from the formulas \xi_i in
%% the previous proof, and why.
%%

The formulas~$\xi_i$ that we defined above are very similar to their
counterparts in the proof of Theorem~\ref{thm:pspace-hardness}---%
and the idea behind their use in the proof is entirely the same
as in the proof of Theorem~\ref{thm:pspace-hardness}.
The only difference is the addition of the
subformulas~$\bigwedge\nolimits_{1 \leq j \leq i} \neg \hat{K}_b \chi'_j$
after the first modal operator.
These additional subformulas are needed to ensure that some
additional worlds---that are a by-product of the combination
of the semi-private announcements~$(\EEE^1_i,f^1_i)$,%
~$(\EEE^2_i,f^3_i)$, and~$(\EEE^3_i,f^5_i)$---%
do not interfere in the reduction.

We show that~$\varphi$ is a true
quantified Boolean formula if and only if~$\MMM,w_0 \models \xi$.
In order to do so, as in the proof of Theorem~\ref{thm:pspace-hardness},
we prove the following (technical) statement relating truth
assignments~$\alpha$ to the propositions~$x_1,\dotsc,x_n$
and (particular) worlds~$w$ in the epistemic model~$(\MMM',w'_0) =
(\MMM,w_0) \otimes (\EEE^1_1,f^1_1) \otimes \dotsm
\otimes (\EEE^3_n,f^5_n)$.
Before we give the statement that we will prove, we observe that
every world~$w$ that sets both~$z_1$ and~$z_2$
to true is the central world of some
group of worlds that represents a truth assignment~$\alpha$
to the propositions~$x_1,\dotsc,x_n$.
For the sake of convenience, we will say
that~$w$ corresponds to the truth assignment~$\alpha$.
The statement that we will prove for all~$1 \leq i \leq n+1$ is the following.

\medskip
\noindent \textit{Statement:}
Let~$\alpha$ be any truth assignment to the propositions~$x_1,\dotsc,x_{i-1}$.
Moreover, let~$w$ be any world in the model~$(\MMM',w'_0)$ such that:
\begin{enumerate}
  \item $w$ makes~$z_1$ and~$z_2$ true,
  \item $w$ makes~$\hat{K}_{b} \chi'_{j}$ false for all~$1 \leq j < i$,
  \item $w$ makes~$\hat{K}_{b} \chi'_{j}$ true
    for all~$i \leq j \leq n$, and
  \item the truth assignment corresponding to~$w$ agrees
    with~$\alpha$ on the propositions~$x_1,\dotsc,x_{i-1}$.
\end{enumerate}
Then the (partially) quantified Boolean formula~$Q_i x_i \dotsc \exists x_{n-1}
\forall x_n. \psi$ is true under~$\alpha$ if and only
if~$w$ makes~$\xi_i$ true.
\medskip

Observe that for~$i = 1$, the world~$w'_0$ satisfies all four conditions.
Therefore, the statement for~$i = 1$ implies
that~$\MMM,w_0 \models \xi$
if and only if~$\varphi$ is a true quantified Boolean formula.
Thus, proving this statement for all~$1 \leq i \leq n+1$ suffices
to show the correctness of our reduction.

We begin by showing that the statement holds for~$i = n + 1$.
In this case, we know that~$\alpha$ is a truth assignment to the
propositions~$x_1,\dotsc,x_n$.
Moreover,~$\xi_{n+1} = \psi'$.
%Because the model~$\MMM'$ is formed using product update
%with the event models~$\EEE^2_n$ and~$\EEE^3_n$,
%we know that~$w$ makes~$\hat{K}_a \chi_n$ true.
%However,~$x_n$ does not occur in~$\psi'$.
%Therefore,
By construction of~$\psi'$, we know
that~$w$ makes~$\psi'$ true if and only if~$\alpha$ satisfies~$\psi$.
Therefore, the statement holds for~$i = n + 1$.

Next, we let~$1 \leq i \leq n$ be arbitrary,
and we assume that the statement holds for~$i+1$.
That is, the statement holds for every combination of a truth
assignment~$\alpha$ and a world~$w$ that satisfies the conditions.
Since~$w$ is a world in the model~$(\MMM,w_0) \otimes
(\EEE^1_1,f^1_1) \otimes \dotsm \otimes (\EEE^3_n,f^5_n)$,
and since~$w$ makes~$z_1$ and~$z_2$ true,
we know that~$w = (w_0,g_1,g'_1,g''_1,\dotsc,g_n,g'_n,g''_n)$,
for some~$g_1,g'_1,g''_1,\dotsc,g_n,g'_n,g''_n$,
where for each~$1 \leq j \leq n$, it holds
that~$g_j \in \SBs f^1_j,f^2_j \SEs$,~$g'_j \in \SBs f^3_j,f^4_j \SEs$,
and~$g''_j \in \SBs f^5_j,f^6_j \SEs$.
%\red{[this needs some fixing]}
%We know that~$w$ makes~$\hat{K}_b \chi'_{j}$ true for all~$i \leq j \leq n$.
%Therefore, we know that for each~$i \leq j \leq n$
%it holds that~$e'_j \in \SBs f^1_j, f^2_j, f^3_j \SEs$.

We now distinguish two cases: either (1)~$i$ is odd,
or (2)~$i$ is even.
In case~(1), the $i$-th
quantifier of~$\varphi$ is existential,
and in case~(2), the $i$-th quantifier of~$\varphi$
is universal.
First, consider case~(1).
Suppose that~$\exists x_i \dotsc \exists x_{n-1} \forall x_n. \psi$ is true
under~$\alpha$. Then there exists some truth assignment~$\alpha'$
to the propositions~$x_1,\dotsc,x_i$ that agrees with~$\alpha$
on the propositions~$x_1,\dotsc,x_{i-1}$
and that ensures that~$\forall x_{i+1} \dotsc
\exists x_{n-1} \forall x_n. \psi$ is true under~$\alpha'$.
Suppose that~$\alpha'(x_i) = 0$; the case for~$\alpha'(x_i) = 1$
is entirely similar.
Now, consider the
worlds~$w' = (w_0,g_1,\dotsc,g''_{i-1},f^2_i,g'_i,g''_i,g_{i+1},\dotsc,g''_n)$
and~$w'' = (w_0,g_1,\dotsc,g''_{i-1},f^2_i,g'_i,f^6_i, \allowbreak g_{i+1},\dotsc,g''_n)$
By the construction of~$\EEE^1_i$,~$\EEE^2_i$, and~$\EEE^3_i$,
and by the semantics of product update,
it holds that~$(w,w') \in R_b$ and~$(w',w'') \in R_a$.
Also, we know that~$w'$
makes~$\bigwedge\nolimits_{1 \leq j \leq i} \neg\hat{K}_b \chi'_{j}$
true.
Moreover, it is straightforward to verify that~$w''$ satisfies
conditions~(1)--(4), for the truth assignment~$\alpha'$.
Also, we know that~$\forall x_{i+1} \dotsc
\exists x_{n-1} \forall x_n. \psi$ is true under~$\alpha'$.
Therefore, by the induction hypothesis,
we know that~$w''$ makes the formula~$\xi_{i+1}$ true.
It then follows from the definition of~$\xi_i$
that~$w'$ and~$w''$ witness that~$w$ makes~$\xi_i$ true.

Conversely, suppose that~$w$ makes~$\xi_i$ true.
Moreover, suppose that~$w'$ and~$w''$ (as defined above) witness this.
It could also be the case that the worlds~$u'$ and~$u''$
witness this, which are obtained from~$w$
by replacing~$g_i$ by~$f^1_i$,
and by replacing~$g_i$ by~$f^1_i$ and~$g'_i$ by~$f^4_i$,
respectively. The case where~$u'$ and~$u''$ witness
that~$w$ makes~$\xi_i$ true is entirely similar.
(There are also variants of~$w'$ and~$w''$,
and of~$u'$ and~$u''$, that could witness the fact that~$w$
makes~$\xi_i$ true. These variants can be obtained by
replacing~$g_j$,~$g'_j$, and~$g''_j$---for~$i < j \leq n$---%
ensuring that for all~$i < j \leq n$ it holds that
neither (1)~$g_j = f^1_j$ and~$g'_j = f^4_j$
nor (2)~$g_j = f^2_j$ and~$g''_j = f^6_j$.
The following argument is entirely similar for these variants.
Therefore, we restrict our attention to the worlds~$w'$ and~$w''$.)
%%
%% NOTE: the above argument about why we don't need to look
%% at these variants in detail is a bit sketchy; we could improve
%% the proof by explaining this in more detail; alternatively, we
%% could modify the construction a bit to make the argumentation
%% at this point easier (and less technically involved).
%%
The assumption that~$w'$ and~$w''$ witness that~$w$
makes~$\xi_i$ true implies that~$w'$
makes~$\bigwedge\nolimits_{1 \leq j \leq i} \neg\hat{K}_b \chi'_{j}$ true
and that~$w''$ makes~$\xi_{i+1}$ true.
Then, by the induction hypothesis, it follows that the truth
assignment~$\alpha'$ to the propositions~$x_1,\dotsc,x_i$
corresponding to the world~$w''$
satisfies~$\forall x_{i+1} \dotsc \exists x_{n-1} \forall x_n. \psi$.
Moreover, since~$\alpha'$ agrees with~$\alpha$ on the
propositions~$x_1,\dotsc,x_{i-1}$, it follows
that~$\exists x_{i} \dotsc \exists x_{n-1} \forall x_n. \psi$
is true under~$\alpha$.

Next, consider case~(2).
Suppose that~$\forall x_i \dotsc \exists x_{n-1} \forall x_n. \psi$
is true under~$\alpha$.
Then for both truth assignments~$\alpha'$
to the variables~$x_1,\dotsc,x_i$ that agree with~$\alpha$
it holds that~$\exists x_{i+1} \dotsc
\exists x_{n-1} \forall x_n. \psi$ is true under~$\alpha'$.
We need to look at those worlds that satisfy~$(z_1 \wedge z_2 \wedge
\bigwedge\nolimits_{1 \leq j \leq i} \neg\hat{K}_b \chi'_{j} \wedge
\bigwedge\nolimits_{i < j \leq n} \hat{K}_b \chi'_{j})$
and that are accessible from~$w$
by a $b$-relation followed by an $a$-relation
(where the intermediate world
makes~$\bigwedge\nolimits_{1 \leq j \leq i} \neg\hat{K}_b \chi'_{j}$ true).
For our argument, it suffices to look at the worlds~$u_1$ and~$u_2$,
where~$u_1 = (w_0,g_1,\dotsc,g''_{i-1},f^1_i,f^4_i,g''_i,g_{i+1},\dotsc,g''_n)$
and~$w'' = (w_0,g_1,\dotsc,g''_{i-1},f^2_i,g'_i,f^6_i,g_{i+1},\dotsc,g''_n)$.
(As in the argument for case~(1) above, there are variants of these
worlds that also satisfy the requirements.
The argument for these variants is entirely similar, and therefore we
restrict our attention to the worlds~$u_1$ and~$u_2$.)
%%
%% NOTE: see above.
%%
The truth assignments~$\alpha_1$
and~$\alpha_2$ that correspond to~$u_1$ and~$u_2$,
respectively,
agree with~$\alpha$ on the propositions~$x_1,\dotsc,x_{i-1}$.
Because~$\forall x_i \dotsc \exists x_{n-1} \forall x_n. \psi$
is true under~$\alpha$, we know
that~$\exists x_{i+1} \dotsc \exists x_{n-1} \forall x_n. \psi$
is true under both~$\alpha_1$ and~$\alpha_2$.
Then, by the induction hypothesis
it follows that both~$u_1$ and~$u_2$ make~$\xi_{i+1}$ true.
Therefore,~$w$ makes~$\xi_i$ true.

Conversely, suppose that~$w$ makes~$\xi_i$ true.
By the definition of~$\xi_i$, we then know that all worlds that
are accessible from~$w$ by a $b$-relation followed by
an $a$-relation and that make~$(z_1 \wedge z_2 \wedge
\bigwedge\nolimits_{1 \leq j \leq i} \neg\hat{K}_b \chi'_{j} \wedge
\bigwedge\nolimits_{i < j \leq n} \hat{K}_b \chi'_{j})$ true
(where the intermediate world
makes~$\bigwedge\nolimits_{1 \leq j \leq i} \neg\hat{K}_b \chi'_{j}$ true),
also make~$\xi_{i+1}$ true.
Consider the worlds~$u_1$ and~$u_2$, as defined above.
These are both accessible from~$w$ by a $b$-relation followed
by an $a$-relation (where the intermediate world
makes~$\bigwedge\nolimits_{1 \leq j \leq i} \neg\hat{K}_b \chi'_{j}$ true),
and they make~$(z_1 \wedge z_2 \wedge
\bigwedge\nolimits_{1 \leq j \leq i} \neg\hat{K}_b \chi'_{j} \wedge
\bigwedge\nolimits_{i < j \leq n} \hat{K}_b \chi'_{j})$ true.
Therefore, both~$u_1$ and~$u_2$ make~$\xi_{i+1}$ true.
Also, the truth assignments~$\alpha_1$ and~$\alpha_2$
that correspond to~$u_1$ and~$u_2$, respectively,
agree with~$\alpha$ on the propositions~$x_1,\dotsc,x_{i-1}$.
Moreover, the truth assignments~$\alpha_1$ and~$\alpha_2$
are both possible truth assignments to the
propositions~$x_1,\dotsc,x_i$ that agree with~$\alpha$.
By the induction hypothesis, the
formula~$\exists x_{i+1} \dotsc \exists x_{n-1} \forall x_n. \psi$ is true
under both~$\alpha_1$ and~$\alpha_2$.
Therefore, we can conclude
that~$\forall x_i \dotsc \exists x_{n-1} \forall x_n. \psi$
is true under~$\alpha$.

This concludes the inductive proof of the statement for
all~$1 \leq i \leq n+1$,
and thus concludes our correctness proof.
Therefore, we can conclude that the problem is \PSPACE{}-hard.
\qed\end{proof}

{
\section{Discussion}
\label{sec:discussion}

In this section, we reflect on the relevance and significance of our
results in the overall endeavor of obtaining a well-informed and useful
understanding of the computational properties of the model checking
problem for DEL.
In particular, we discuss (a)~how our results contribute to the undertaking of
getting a detailed picture of the computational complexity of DEL model checking,
(b)~why such a detailed theoretical picture is useful and important for the
development and improvement of DEL model checking algorithms,
and (c)~how we can get an even more detailed picture of the complexity of DEL
model checking in future research.

\paragraph{Detailed worst-case complexity analysis}
In this paper---as in most works in the literature on the study
of the computational properties of DEL---%
we adopt the framework of worst-case computational complexity
analysis (see, e.g.,~\cite{AroraBarak09}).
This is a framework that has been hugely influential and successful,
but that also has its inherent downsides.
One of its main disadvantages is that it is prone to give an overly
negative image of the computational difficulty of a problem.
It is not uncommon for a problem to be computationally hard in the
worst case sense, while instances of this problem that come up in
applications can be solved efficiently.
Therefore, for a worst-case computational complexity analysis to provide
an accurate picture of the inherent complexity of a problem, it needs
to be as fine-grained and detailed as possible.
This means that it needs to consider many different restricted settings
that are relevant to applications that use the problem under study.

The results that we provided in Sections~\ref{sec:updates-with-event-models}
and~\ref{sec:semi-private-announcements} contribute to
the detail and fine-grainedness of the computational complexity study
of the model checking problem for DEL.
Previous work on the computational complexity of the problem
investigated restricted settings---where certain components of the
problem are restricted in number or shape.
However, this mostly involved studies where restrictions only involve
a single component of the problem
\cite{AucherSchwarzentruber13,VanBenthemVanEijckKooi06,%
VanEijckSchwarzentruber14,KooiVanBenthem04}.
There has been some research that considered combinations of restrictions
\cite{VandePol15,VandePolVanRooijSzymanik18}, but this work
only focuses on the line between polynomial-time solvable
(and an extension thereof: fixed-parameter tractability) and
computationally intractable.
Our results both (i)~take into account combinations of restrictions on different
components of the problem, and (ii)~are aimed at identifying the exact degree
of complexity of the problem---distinguishing between different degrees
of computational intractability.
As such, our results provide an important and useful step in the direction of
establishing a detailed and fine-grained picture of the worst-case complexity
of the model checking problem for DEL.

\paragraph{Guidance for model checking algorithms}
Establishing a detailed picture of the exact degree of computational complexity
for a wide range of restricted settings provides a good guide for the development
of practical model checking algorithms.
For example, the results of Theorems~\ref{thm:pspace-hardness}
and~\ref{thm:semi-private-pspace-hardness} show that the model checking
problem can require polynomial space even in restricted settings with a limited
number of agents and propositions.
This suggests that better performing algorithms could
be obtained by using optimized algorithmic approaches for PSPACE-complete
problems.
For example, it would be interesting to investigate whether encoding the
model checking problem for DEL into the satisfiability problem for
quantified Boolean formulas (QBFs) and subsequently invoking QBF solvers
on the resulting formula could lead to model checking algorithms that
are competitive with existing model checking algorithms---even on instances
that involve only a limited number of agents and propositions.
Such an approach would have the benefit that years of research and
engineering effort on developing QBF solvers
(see, e.g., \cite{GiunchigliaMarinNarizzano09}) could be leveraged
to get efficient algorithms.

The results that we established in this paper indicate that currently implemented
model checking algorithms for DEL---that are based on constructing a
representation of the epistemic model resulting from the original model
and updates applied to it---are likely to run into barriers of combinatorial
explosion already in very limited settings.
For example, the results of Theorems~\ref{thm:delta2},%
~\ref{thm:pspace-hardness2} and~\ref{thm:pspace-hardness} show that
deviating from the restricted setting of Proposition~\ref{prop:ptime}
in one of various minimal ways leads to a setting where any
(deterministic) algorithm cannot run in polynomial time in the worst case.
Examples of model checking algorithms for which these results are relevant
are those of DEMO \cite{VanEijck07}
and SMCDEL \cite{VanBenthemVanEijckGattingerSu18,Gattinger18}.

Our results also suggest directions for experimental evaluation
of (implemented) model checking algorithms.
For example, it would be useful to investigate on which types of instances
algorithms such as DEMO and SMCDEL perform well,
and on which types of instances they in fact run into barriers of combinatorial
explosion.
For instances where DEMO and SMCDEL perform poorly, it would be interesting
to study whether model checking algorithms
based on QBF solvers---and other optimized algorithmic methods for PSPACE-complete
problems---perform better.
The computational complexity results in this paper provide indications
for which properties of inputs could have an impact on the performance of different
model checking algorithms.

\paragraph{Parameterized complexity analysis}
The foundational computational complexity results for DEL model checking that
we developed in this paper (and that other papers in the literature developed)
enable an interesting direction for future research---%
namely, to employ the framework of \emph{parameterized
complexity theory} (see, e.g.,~\cite{CyganEtAl15,DowneyFellows13,%
FlumGrohe06}). This would take the undertaking
of providing a more fine-grained worst-case complexity analysis even a
step further.
Parameterized complexity theory
provides a complexity-theoretic framework that can be used to identify which
parts of the problem input contribute in what way to the running time
(or space usage) of algorithms solving the model checking problem for DEL.
This framework has already been used to initiate a more detailed investigation
of the computational complexity of the model checking problem for DEL
\cite{VandePol15,VandePolVanRooijSzymanik18}.
Further pursuing this research direction has the potential of yielding useful
and relevant insights into the computational properties of DEL.

The results in this paper provide a constructive foundation for establishing
parameterized complexity results for the model checking problem for DEL.
For example, the hardness result of Theorem~\ref{thm:pspace-hardness}
tells us that the model checking problem for DEL is
\para{\PSPACE}-complete when parameterized by the number of agents
and the number of propositional variables occurring in the formula---%
and thus is not fixed-parameter tractable for this parameter.
(For more details on the relation between traditional computational complexity
results and parameterized complexity, we refer to the literature---%
e.g.,~\cite{FlumGrohe03,FlumGrohe06}.)

\paragraph{Studying other restrictions}
Another way forward in the study of the computational properties of the model
checking problem for DEL that is pointed at by the results that we provide in this
paper, is to consider restrictions on the problem input that go beyond counting
simple quantities in the input (such as the number of agents or the number of
propositions) and instead are based on structural properties of the input.
An example of this would be to consider the computational properties
of settings where models are restricted to those whose underlying graph
has certain graph-theoretic properties---such as bounded treewidth.%
\footnote{{Treewidth is a measure that,
intuitively, captures how similar a
graph is to a tree (trees have minimum possible treewidth). Restricting problems
to graphs of bounded treewidth often yields efficient algorithms
for problems that are intractable in general
(see, e.g., \cite[Chapter~7]{CyganEtAl15}).}}
Whereas simple restrictions only lead to positive algorithmic results in only
a very limited number of cases---as indicated by the results that we provide in this
paper---structural restrictions have the potential of leading to positive results
in more general settings.
Tractability results based on such structural properties could then of course be used
to develop efficient model checking algorithms that are tailored to application
settings where these structural properties show up in problem inputs.
The framework of parameterized complexity theory is particularly well suited
to analyze the impact of structural properties of the problem input on the computational
complexity of the model checking problem for DEL.
}

%%%
%%% CONCLUSION
%%%
\section{Conclusion}
\label{sec:conclusion}

%\todo{Update, according to revision (at the end).}
{
We extended the investigation of the computational complexity
of the model checking problem for DEL
by providing a detailed computational complexity analysis
of the model checking problem for various (previously uninvestigated)
combinations of restrictions on the DEL model.
}
In particular, we studied various restrictions of the problem
where all models are S5,
including bounds on the number of agents, allowing only
single-pointed models, allowing no postconditions,
and allowing semi-private announcements rather than
updates with arbitrary event models.
We showed that the problem is already \PSPACE{}-hard
for very restricted settings.

Future research includes extending the computational
complexity analysis to additional restricted settings.
For instance, it would be interesting to see whether
the polynomial-time algorithm for Proposition~\ref{prop:ptime}
can be extended to the setting where the models
contain only relations that are transitive, Euclidean and serial
(KD45 models).
{
Additionally, it would be interesting to further investigate the
contribution of various parameters of the problem input to the
computational costs required to solve the problem---%
continuing an endeavor that was recently initiated
\cite{VandePol15,VandePolVanRooijSzymanik18}.
}
In the setting of KD45 models,
it would also be interesting to investigate the complexity of
the problem for
the case where all updates are private announcements
(i.e., a public announcement to a subset of agents,
where the remaining agents have no awareness that any action
has taken place).
Moreover, future research includes obtaining upper
bounds for the case where we only found lower bounds
(i.e., for the case of one agent, a single-pointed models,
and single-pointed event models with postconditions,
where we showed \DeltaP{2}-hardness).

%%%
%%% ACKNOWLEDGMENTS
%%%
\begin{acknowledgements}
%Ronald de Haan was supported by
%the Austrian Science Fund (FWF),
%project P26200 (Parameterized Compilation).
%
We would like to thank anonymous reviewers for
their useful feedback on previous versions of the paper.
\end{acknowledgements}

\bibliographystyle{spmpsci}      % mathematics and physical sciences
%\bibliography{DEL}   % name your BibTeX data base

\newcommand{\noopsort}[1]{}

\end{document}